\documentclass[a4paper,12pt,reqno]{article}
\usepackage{a4wide,amssymb,amsmath,amsthm}
\usepackage{color,mathrsfs}
\usepackage{todonotes}
\usepackage{hyperref}
\usepackage{tikz}
\usetikzlibrary{shapes,arrows}

\newcommand{\beq}{\begin{equation}}
\newcommand{\eeq}{\end{equation}}

\newcommand{\ben}{\begin{enumerate}}
\newcommand{\een}{\end{enumerate}}

\newcommand{\M}{\mathcal{M}}
\newcommand{\N}{\mathcal{N}}
\newcommand{\Ms}{\textbf{\textit{M}}}
\newcommand{\Ns}{\textbf{\textit{N}}}
\newcommand{\g}{\textbf{\textit{g}}}
\newcommand{\h}{\textbf{\textit{h}}}
\newcommand{\Mgot}{(\M,\g,\mathfrak{o},\mathfrak{t})}

\newcommand{\EM}{E_\Ms}
\newcommand{\PM}{P_\Ms}
\newcommand{\EMh}{E_{\Ms[\h]}}
\newcommand{\PMh}{P_{\Ms[\h]}}
\newcommand{\EMsh}{E_{\Ms[s\h]}}
\newcommand{\PMsh}{P_{\Ms[s\h]}}
\newcommand{\A}{\mathscr{A}}
\newcommand{\Ak}{\mathscr{A}^{\text{\textnormal{kin}}}}
\newcommand{\Ad}{\mathscr{A}^{\text{\textnormal{dyn}}}}
\newcommand{\Ab}{\A^\bullet}
\newcommand{\W}{\mathscr{W}}
\newcommand{\Wb}{\mathscr{W}^\bullet}
\newcommand{\Wk}{\mathscr{W}^{\text{\textnormal{kin}}}}
\newcommand{\Wd}{\mathscr{W}^{\text{\textnormal{dyn}}}}
\newcommand{\wkin}{\omega^\text{\textnormal{kin}}}
\newcommand{\wdyn}{\omega^\text{\textnormal{dyn}}}
\newcommand{\F}{\mathscr{F}}
\newcommand{\T}{\mathscr{T}}
\newcommand{\J}{\mathscr{J}}
\newcommand{\JJ}{\widetilde{\mathscr{J}}}
\newcommand{\K}{\mathscr{K}}
\renewcommand{\H}{\mathscr{H}}
\newcommand{\rce}{\operatorname{rce}}
\newcommand{\tenpow}[2]{{#1}^{\otimes#2}}
\renewcommand{\bar}{\overline}
\newcommand{\supp}{\operatorname{supp}}
\newcommand{\cadv}{\chi^{\text{\textnormal{adv}}}}
\newcommand{\cret}{\chi^{\text{\textnormal{ret}}}}
\newcommand{\Sadv}{\Sigma^\text{\textnormal{adv}}}
\newcommand{\Sret}{\Sigma^\text{\textnormal{ret}}}
\newcommand{\Cs}{C^\infty}
\newcommand{\Ccs}{\Cs_0}
\newcommand{\id}{\mathbf{1}}
\newcommand{\sM}{\sim_\Ms}
\newcommand{\Loc}{\textnormal{\textsf{Loc}}}
\newcommand{\Alg}{\textnormal{\textsf{Alg}}}
\newcommand{\dds}{\frac{d}{ds}}
\newcommand{\Zeta}{Z}
\newcommand{\tl}{\widetilde{\lambda}}
\newcommand{\cH}{\check{H}}
\newcommand{\hEM}{\hat{E}_\Ms}
\newcommand{\cEM}{\check{E}_\Ms}

\newtheorem{theorem}{Theorem}
\newtheorem{corollary}[theorem]{Corollary}
\newtheorem{lemma}[theorem]{Lemma}

\numberwithin{theorem}{section}
\numberwithin{equation}{section}

\begin{document}

\title{Dynamical locality of the nonminimally coupled scalar field and enlarged algebra of Wick polynomials}
\author{Matthew Ferguson\footnote{\textsf{mtf500@york.ac.uk}}\\Department of Mathematics, University of York}
\date{March 9, 2012}
\maketitle

\begin{abstract}
We discuss dynamical locality in two locally covariant quantum field theories, the nonminimally coupled scalar field and the enlarged algebra of Wick polynomials. We calculate the relative Cauchy evolution of the enlarged algebra, before demonstrating that dynamical locality holds in the nonminimally coupled scalar field theory. We also establish dynamical locality in the enlarged algebra for the minimally coupled massive case and the conformally coupled massive case.
\end{abstract}

\section{Introduction}
\label{sect:intro}

The concept of Axiomatic Quantum Field Theory has traditionally been explored only in Minkowski space: in particular, the Wightman axioms \cite{pctwightman} and the Haag-Kastler axioms \cite{haaglocalqf} outline ways of providing a set of axioms for a quantum field theory to obey. Over the past decade, advances have been made in the area of Axiomatic Quantum Field Theory in curved spacetimes. In particular, the work by Brunetti, Fredenhagen and Verch \cite{bfv} outlined a set of axioms, similar to the Haag-Kastler axioms for Quantum Field Theory on Minkowski space, that should be obeyed by any QFT that can be defined on curved spacetimes. The Haag-Kastler axiomatic framework is often described as \emph{Algebraic Quantum Field Theory}; the axioms lay out certain properties that should be held by any legitimate assignment of an algebra of observables to each arbitrary region of Minkowski space. Extending algebraic QFT to curved spacetime involves examining the ways one might amend these axioms to define the properties held by a suitable assignment of an algebra of observables to arbitrary regions of arbitrary spacetimes. In practice, though, to achieve meaningful results we have to apply some restrictions to the type of region and the type of spacetime we are allowed to choose. The axioms proposed in \cite{bfv} use the tools of category theory; the allowed regions in this case are open, globally hyperbolic subregions of globally hyperbolic spacetimes (definitions are given in section \ref{sect:lcdl}), and allow us to think of a particular quantum field theory as a functor between the category \Loc\ whose objects are globally hyperbolic spacetimes, and the category \Alg\ whose objects are at the very least $\ast$-algebras, but may possibly possess additional structure.

	However, it turns out that these axioms alone allow for some rather undesirable pathological theories. In particular, some very recent work by Fewster and Verch \cite{dynloc} has shown that certain theories may satisfy the BFV axioms despite in some sense describing different physics depending on the spacetime to which it assigns the algebra of observables. The question of precisely what is meant by a theory representing the same physics in all possible spacetimes is still, by and large, an open one. While it is desirable to find a condition on theories that somehow formalises this property, this question is more easily answered by comparing theories with one other, and so the  \emph{SPASs} (Same Physics in All Spacetimes) condition proposed in the aforementioned paper is a condition on \emph{classes} of theories. It is intended to be a necessary condition for such a class to comprise theories, 
each of which represents the same physics in all spacetimes, according to some common definition of the term. In this paper we are concerned with the class of \emph{dynamically local} theories, which is shown in \cite{dynloc} to satisfy the SPASs condition.

The property of dynamical locality has other desirable consequences such as a no-go theorem for natural states, and dynamical locality has so far been demonstrated for some linear theories, including the minimally coupled massive Klein-Gordon field and the massless current algebra. However, it fails in the case of the minimally coupled massless Klein-Gordon field. It is therefore desirable to find further examples of well-known theories that can be constructed in a locally covariant way that either satisfy or violate dynamical locality.
We will prove in this paper that the nonminimally coupled Klein-Gordon scalar field is dynamically local in both the massive and massless case, and also that the extended algebra of noninteracting Wick polynomials can be shown to be dynamically local in the minimally coupled massive and conformally coupled massive cases; however, it fails to be dynamically local in the minimally coupled massless case.


\section{Local covariance and Dynamical Locality}
\label{sect:lcdl}

We are using the prescription in \cite{bfv} for the construction of locally covariant theories, in which a theory is considered to be functor from a category of spacetimes to a category of algebras. We must therefore first define the categories we will be using. We will follow the definitions and notation in \cite{dynloc} for the category of globally hyperbolic spacetimes. This category is denoted \Loc; its objects are quadruples $\Ms=\Mgot$ where $\M$ is a smooth paracompact orientable nonempty $n$-dimensional manifold with finitely many connected components, $\g$ is a smooth time-orientable metric for $\M$ with signature $+{}-\cdots-{}$, and $\mathfrak{o}$ and $\mathfrak{t}$ are choices of orientation and time-orientation respectively for $\M$. These spacetimes must also satisfy global hyperbolicity: there can be no closed causal curves in $\M$, and for each pair $p,q\in\M$ the intersection $J_\Ms^-(p)\cap J_\Ms^+(q)$ must be compact, where $J_\Ms^\pm(p)$ denotes the causal future ($+$) or past ($-$) of $p$ in $\M$.

An arrow of \Loc\ from an object $\Ns=(\N,\g_\N,\mathfrak{o}_\N,\mathfrak{t}_\N)$ to a second object $\Ms=(\M,\g_\M,\mathfrak{o}_\M,\mathfrak{t}_\M)$ is a smooth embedding $\psi:\N\hookrightarrow\M$ that is isometric (i.e.\ $\psi^\ast\g_\M=\g_\N$) and orientation- and time-orientation-preserving (i.e.\ $\psi^\ast\mathfrak{o}_\M=\mathfrak{o}_\N$, $\psi^\ast\mathfrak{t}_\M=\mathfrak{t}_\N$). It must also respect the causal structure: the image $\psi(\N)\subseteq\M$ must be \emph{causally convex} in $\M$, i.e.\ each causal curve in $\M$ with both endpoints lying within $\psi(\N)$ must be entirely contained within $\psi(\N)$.

A \emph{Cauchy surface} $\Sigma$ for a spacetime $\Ms$ is a subset of $\M$ that is intersected by every inextendible timelike curve in $\M$ exactly once. Clearly no Cauchy surface can have a timelike tangent at any point, but this definition does allow a Cauchy surface to have a null tangent; consequently we will refer to a Cauchy surface whose tangents are all spacelike as a \emph{spacelike Cauchy surface}. Global hyperbolicity of $\Ms$ is equivalent to $\Ms$ containing a smooth spacelike Cauchy surface \cite{bernalsanchez03}. An arrow in $\Loc$ whose image contains a Cauchy surface for its target is called a \emph{Cauchy arrow}. We may safely blur the distinction between a spacetime and its underlying manifold, so in the remainder of this paper we may use the same notation (e.g.\ $\Ms$) for both; for example, we will denote by $\Ccs(\Ms)$ the space of compactly supported smooth functions on the underlying manifold $\M$.

The category whose objects are candidates for the algebras of observables of a theory is denoted \Alg. The objects of \Alg\ are unital $\ast$-algebras, and the morphisms are unit-preserving $\ast$-monomorphisms.


\subsection{Locally covariant theories}
\label{sect:lc}
A \emph{locally covariant quantum field theory} is defined to be a covariant functor from \Loc\ to \Alg\ \cite{bfv}. That is, a theory $\A$ maps objects of \Loc\ to objects of \Alg, and arrows of \Loc\ to arrows of \Alg, such that:
\begin{itemize}
	\item for any \Loc-arrow $\psi:\Ns\hookrightarrow\Ms$, the arrow $\A(\psi)$ has domain $\A(\Ns)$ and codomain 
	$\A(\Ms)$,
	\item for any two \Loc-arrows $\psi_1:\textbf{\textit{O}}\hookrightarrow\Ns,$ $\psi_2:\Ns\hookrightarrow\Ms$, 
	we have $\A(\psi_2\circ\psi_1)=\A(\psi_2)\circ\A(\psi_1)$,
	\item for any spacetime $\Ms$, we have $\A(\text{id}_\Ms)=\text{id}_{\A(\Ms)}$.
\end{itemize}
While this is the only property a theory needs to satisfy to be locally covariant, we generally wish to apply some further conditions on the theories we work with. In particular, there is no condition pertaining to causality in the basic definition of a locally covariant theory. A locally covariant theory $\A$ is said to be \emph{causal} if it has the following property: let $\psi_1:\Ns_1\hookrightarrow\Ms,\ \psi_2:\Ns_2\hookrightarrow\Ms$ be arrows in \Loc\ such that the images $\psi_1(\Ns_1)$ and $\psi_2(\Ns_2)$ are causally disjoint in $\Ms$. Then $[\A(\psi_1)A_1,\A(\psi_2)A_2]=0$ for any $A_1\in\A(\Ns_1)$, $A_2\in\A(\Ns_2)$.

We will also generally require our theories to satisfy the \emph{timeslice axiom}. Suppose $\psi:\Ns\hookrightarrow\Ms$ is an arrow in \Loc; a locally covariant theory $\A$ obeys the timeslice axiom if the \Alg-arrow $\A(\psi)$ is an isomorphism whenever the image of $\Ns$ in $\Ms$ under $\psi$ contains a Cauchy surface for $\Ms$ (alternatively, $\A(\psi)$ is an isomorphism whenever $\psi$ is a Cauchy arrow). The timeslice axiom allows us to define an automorphism of an algebra $\A(\Ms)$ called the relative Cauchy evolution, which is defined as follows.

For any spacetime $\Ms=\Mgot$, we define $\h\in\Ccs(T^0_2\M)$ to be a \emph{metric perturbation} if it is symmetric and the spacetime $\Ms[\h]=(\M,\g+\h,\mathfrak{o},\mathfrak{t}')$ is also an object in \Loc\ (where $\mathfrak{t}'$ is the unique choice of  time-orientation that coincides with $\mathfrak{t}$ outside $\supp(\h)$). The set of all such metric perturbations on $\Ms$ is denoted $H(\Ms)$, and for any $O\subset\M$ we denote by $H(\Ms;O)$ all $\h\in H(\Ms)$ whose support lies within $O$.

Given some $\h\in H(\Ms)$, we pick globally hyperbolic subregions $\N^\pm$ of $\M$ such that each contains a Cauchy surface $\Sigma^\pm$ for $\Ms$, and such that $\N^\pm\subseteq\M\setminus J_\Ms^\mp(\supp(\h))$. Now, we can consider the spacetimes $\Ns^\pm=(\N^\pm,\g|_{\N^\pm},\mathfrak{o}|_{\N^\pm},\mathfrak{t}|_{\N^\pm})$ in their own right; each is a sub-spacetime of both $\Ms$ and $\Ms[\h]$, and we denote by $\iota^\pm,\iota^\pm[\h]$ the canonical embeddings of $\Ns^\pm$ respectively into $\Ms$ and $\Ms[\h]$. If a locally covariant theory $\A$ satisfies the timeslice axiom, then the arrows $\A(\iota^\pm)$ and $\A(\iota^\pm[\h])$ must be isomorphisms. It follows that we can form an automorphism $\rce_\Ms[\h]$ on $\A(\Ms)$ defined by
\[
	\rce_\Ms[\h]=\A(\iota^-)\circ\A(\iota^-[\h])^{-1}\circ\A(\iota^+[\h])\circ\A(\iota^+)^{-1},
\]
called the \emph{relative Cauchy evolution} on $\Ms$ induced by $\h$. The relative Cauchy evolution can be shown to be independent of the choice of future and past subspacetimes $\Ns^\pm$ \cite[Prop.\ 3.3]{dynloc}.


\subsection{Dynamical locality}
\label{sect:dynloc}

It is natural to ask the question of whether the condition of local covariance, with the timeslice axiom, is enough to ensure that a theory is ``physically realistic''. As discussed before, the existence of certain pathological locally covariant theories has motivated the discussion in \cite{dynloc}, where the idea of the Same Physics in All Spacetimes (SPASs) is introduced as a condition on classes of theories that is claimed to be necessary for the theories to be considered physically realistic. A class of theories $T$ has the SPASs property if, whenever
\begin{itemize}
	\item $\A,\mathscr{B}\in T$,
	\item there exists a natural transformation $\zeta:\A\stackrel\cdot\longrightarrow \mathscr{B}$, and
	\item there exists a globally hyperbolic spacetime $\Ms$ on which $\zeta_\Ms$ is an isomorphism,
\end{itemize}
then $\zeta$ is a natural isomorphism (i.e.\ $\zeta_\Ns$ is an isomorphism for each globally hyperbolic spacetime $\Ns$).
It is shown in \cite{dynloc} that one can construct a class $T$ of locally covariant causal theories that obey the timeslice axiom, but such that $T$ does not have the SPASs property. To this end, it is suggested that the additional axiom of \emph{dynamical locality}, defined below, is imposed. The class of dynamically local theories is a subclass of the locally covariant theories that obey the timeslice axiom, but it has the added advantage of satisfying the SPASs condition.

We first define the \emph{kinematic nets} and \emph{dynamical nets} of a locally covariant theory $\A$ obeying the timeslice axiom. Let $\Ms$ be a globally hyperbolic spacetime and $O$ be a globally hyperbolic open subregion of $\Ms$ with finitely many connected components, all of which are causally disjoint (we denote by $\mathscr{O}(\Ms)$ the set of possible such $O$). Clearly we can regard $\Ms|_O$ as a globally hyperbolic spacetime in its own right. We will denote the map embedding $\Ms|_O$ into $\Ms$ by $\iota_{\Ms;O}$. When we apply the functor $\A$ to $\Ms|_O$, we get the algebra $\A(\Ms|_O)$, which can be embedded in $\A(\Ms)$ by the map $\alpha^{\text{kin}}_{\Ms;O}$, defined to be the result of applying the same functor to $\iota_{\Ms;O}$. The \emph{kinematic net} is defined to be the map which assigns $O\mapsto\alpha^\text{kin}_{\Ms;O}$. The algebra obtained by applying $\A$ to the restriction $\Ms|_O$ is called the \emph{kinematic algebra} of $O$, denoted by $\Ak(\Ms;O)=\A(\Ms|_O)$.

Given such $\Ms$ and $O$, we can also define the \emph{dynamical net} as follows: given $O\in\mathscr{O}(\Ms)$, and compact $K\subset O$, we let
\[
	\A^\bullet(\Ms;K)=\{A\in\A(\Ms):\rce_\Ms[\h]A=A\text{ for all }\h\in H(\Ms;K^\perp)\}.
\]
Here $K^\perp=\Ms\setminus J_\Ms(K)$ denotes the causal complement of a compact $K\subseteq\Ms$.
We then define the \emph{dynamical algebra} as
\beq
	\label{eqn:dynalgdef}
	\Ad(\Ms;O)=\bigvee_{K\in\mathscr{K}(\Ms;O)}\A^\bullet(\Ms;K),
\eeq
where $\mathscr{K}(\Ms;O)$ is the set of compact subsets of $\Ms$ with a \emph{multi-diamond} neighbourhood based in $O$. Here a multi-diamond is a finite union of causally disjoint diamonds, where we use the following definition from \cite{brunruzz}: a \emph{diamond} is a set $D_\Ms(B)$ such that there exists a spacelike Cauchy surface $\Sigma\subset\Ms$, and a chart $(U,\phi)$ of $\Sigma$, where $\phi(B)$ is a nonempty open ball in $\mathbb{R}^{n-1}$ with closure contained in $\phi(U)$, and $D_\Ms(B)$ denotes the domain of dependence of $B$.
The inclusion
\[
	\alpha^{\text{dyn}}_{\Ms;O}:\Ad(\Ms;O)\hookrightarrow\A(\Ms)
\]
is unique (up to isomorphism), and we define the dynamical net to be the map which assigns $O\mapsto\alpha_{\Ms;O}^\text{dyn}$.

A theory is defined to be \emph{dynamically local} if for every globally hyperbolic spacetime $\Ms$ and nonempty $O\in\mathscr{O}(\Ms)$, we have $\Ak(\Ms;O)\cong\Ad(\Ms;O)$, or alternatively
\[
	\alpha^\text{dyn}_{\Ms;O}\cong\alpha^\text{kin}_{\Ms;O}.
\]
This is equivalent to demanding that for all such $O,\Ms$ we have 
\[
	\alpha^\text{dyn}_{M;O}(\Ad(M;O))=\alpha^\text{kin}_{M;O}(\Ak(M;O)).
\]
For an additive theory, that is, one in which $\Ak(M;O)$ is generated by its subalgebras
corresponding to relatively compact subregions of $O$, \cite[Prop. 6.1]{dynloc} entails that we always have $\alpha^\text{kin}_{M;O}(\Ak(M;O))\subseteq\alpha^\text{dyn}_{M;O}(\Ad(M;O))$, and therefore it is sufficient for dynamical locality to show that 
\beq
	\label{eqn:dynloccond}
	\alpha^\text{dyn}_{M;O}(\Ad(M;O))\subseteq\alpha^\text{kin}_{M;O}(\Ak(M;O)).
\eeq
This applies to all the theories we will study here.


\section{The Klein-Gordon Field and Wick Polynomials as LCTs}
\label{sect:const}

\subsection{Construction of the Klein-Gordon Theory}
\label{sect:kgfconst}

The Klein-Gordon operator on a spacetime $\Ms$ is denoted $P_\Ms =\Box_\g+\xi R_\g+m^2$. We call any solution $\phi\in\Cs(\Ms)$ to the field equation $P_\Ms \phi=0$ a \emph{classical solution} to the field equation. The coupling constant $\xi\in\mathbb{R}$ and the mass $m\geq0$ are held constant over all spacetimes. The Klein-Gordon operator has associated with it two unique continuous linear operators $E_\Ms ^\pm:\Ccs(\Ms)\to\Cs(\Ms)$ with the properties
\begin{align}
	E_\Ms ^\pm P_\Ms f&=f=P_\Ms E_\Ms ^\pm f\label{eqn:funsol1}\\
	\supp(E_\Ms ^\pm f)&\subseteq J_M^\pm(\supp(f))\label{eqn:funsol2}
\end{align}
for any $f\in\Ccs(\Ms)$ \cite{wald} (here we identify $\EM^\pm\PM f$ and $\PM\EM^\pm f$ with their preimage under the canonical embedding $\iota:\Ccs(\Ms)\hookrightarrow\Cs(\Ms)$). The operator $E_\Ms =E_\Ms ^--E_\Ms ^+$ is the \emph{(advanced-minus-retarded) fundamental solution} for the Klein-Gordon field on $\Ms$, and any classical solution $\phi$ with compact support on Cauchy surfaces is of the form $\phi=E_\Ms f$ for some $f\in\Ccs(\Ms)$. We denote by $E_\Ms (x,y)$ the antisymmetric bidistribution on test functions satisfying
\[
	\int_\Ms dy\,E_\Ms(x,y)f(y)=(E_\Ms f)(x)
\]
for each $f\in\Ccs(\Ms)$. Furthermore, we denote
\[
	E_\Ms (f,f')=\int_{\Ms}dx\,f(x)(\EM f')(x)=\int_{\Ms^{\times 2}}dx\,dy\,f(x)
	\EM(x,y)f'(y),
\]
for $f,f'\in\Ccs(\Ms)$. Note that this entails
\beq
	\label{eqn:Eantiprop}
	\int_\Ms dx\,f(x)(\EM f')(x)=-\int_\Ms dx\,(\EM f)(x)f'(x).
\eeq

Given a fixed spacetime $\Ms$, we can construct the algebra of the Klein-Gordon quantum field theory as the unital $\ast$-algebra generated by elements $\Phi_\Ms (t)$, $t\in\Ccs(\Ms)$ satisfying the following four conditions:
\begin{subequations}
\begin{align}
	&\text{The assignment }t\mapsto\Phi_\Ms (t)\text{ is linear},\label{eqn:kgcond1}\\
	&\Phi_\Ms (t)^\ast=\Phi(\bar t),\label{eqn:kgcond2}\\
	&[\Phi_\Ms (t),\Phi_\Ms (t')]=iE_\Ms (t,t')\id,\label{eqn:kgcond3}\\
	&\Phi_\Ms (P_\Ms t)=0.\label{eqn:kgcond4}
\end{align}
\end{subequations}
While it can be observed that this algebra can be represented simply as a deformation of the symmetric tensor algebra $\Gamma_\odot(E_\Ms \Ccs(\Ms))$ (see e.g.\ \cite{dynloc2}), alternative ways of constructing this algebra can be seen in \cite{bf:qftcb,chilfred}. The following treatment is based on \cite{bf:qftcb}.

If we remove the condition \eqref{eqn:kgcond4}, then the algebra generated by the other three conditions is isomorphic to the unital $\ast$-algebra $\F(\Ms)$ comprising functionals on $\Cs(\Ms)$ of the form
\beq
	\label{eqn:funcform}
	F[f]=\sum_{n=0}^N\int_{\Ms^{\times n}}d^nx\,t_n(x_1,\ldots,x_n)f(x_1)\cdots f(x_n),
\eeq
where each $t_n$ is a totally symmetric finite sum of products of test functions in one variable:
\[
	t_n(x_1,\ldots,x_n)=\mathbf{S}\sum_{j\text{ finite}}
	\prod_{k=1}^n\varphi_{jk}(x_k)
\]
for some $\varphi_{jk}\in\Ccs(\Ms)$, where $\mathbf{S}$ denotes symmetrisation. We denote the set of all such $t_n$ as $\F^n(\Ms)$; we define $\F^0(\Ms)=\mathbb{C}$, and we may note that $\F^1(\Ms)=\Ccs(\Ms)$. We will use the shorthand notation
\[
	t_n[f]=\int_{\Ms^{\times n}}d^nx\,t_n(x_1,\ldots,x_n)f(x_1)\cdots f(x_n).
\]
For each $F=\sum_{n=0}^Nt_n$ with $t_N\neq0$ we define $O(F)=N<\infty$.

The $k^\text{th}$ functional derivative of $F=\sum_{n=0}^Nt_n$ is given by
\beq
	\label{eqn:kgfuncderiv1}
	F^{(k)}[f](x_1,\ldots,x_k)=\sum_{n=k}^Nt^{(k)}_n[f](x_1,\ldots,x_k),
\eeq
where for $k\leq n$,
\beq
	\label{eqn:kgfuncderiv2}
	t^{(k)}_n[f](x_1,\ldots,x_k)=\frac{n!}{(n-k)!}\int_{\Ms^{\times(n-k)}}dx_{k+1}
	\cdots dx_n\,t_n(x_1,\ldots,x_n)f(x_{k+1})\cdots f(x_n).
\eeq
For any $f\in\Ccs(\Ms)$, we may regard the functional derivative $F^{(k)}[f](x_1,\ldots,x_k)$ of an element $F\in\F(\Ms)$ as an element of $\F^k(\Ms)$ for $k\leq O(F)$. Addition in $\F(\Ms)$ is given by addition of functionals, and products of elements are defined by
\beq
	\label{eqn:kgproddef}
	(F\star F')[f]=\sum_{k=0}^{\text{min}(O(F),O(F'))}\frac{i^k}{2^kk!}E_\Ms^k\left(F^{(k)}[f],F'^{(k)}[f]\right),
\eeq
where for $t,t'\in\F^k(\Ms)$, we have
\beq
	\label{eqn:kgproddef2}
	E_\Ms^k(t,t')=\int_{\Ms^{\times(2k)}}d^kx\,d^ky\,t(x_1,\ldots,x_k)t'(y_1,\ldots,y_k)\prod_{j=1}^k E_\Ms(x_j,y_j),
\eeq
and for $\alpha,\beta\in\F^0(\Ms)=\mathbb{C}$,
\[
	E_\Ms^0(\alpha,\beta)=\alpha\beta.
\]
The product \eqref{eqn:kgproddef} can be shown to be associative. The involution of $F=\sum_{n=0}^Nt_n\in\F(\Ms)$ is given by $F^\ast=\sum_{n=0}^N\bar t_n$, and the identity with respect to the $\star$ product is the constant functional $\id[f]\equiv 1$.


The algebra $\F(\Ms)$ is generated by elements satisfying conditions \eqref{eqn:kgcond1}--\eqref{eqn:kgcond3}, so it should be the case that we can recover the algebra $\A(\Ms)$ by reapplying condition \eqref{eqn:kgcond4}. The set $\J(\Ms)$, defined to be the set containing all elements $F\in\F(\Ms)$ satisfying $F[E_\Ms f]=0$ for all $f\in\Ccs(\Ms)$, is a two-sided $\ast$-ideal in $\F(\Ms)$ \cite{bf:qftcb};
on taking the quotient $\F(\Ms)/\J(\Ms)$ we obtain the algebra $\A(\Ms).$ We have the following result:
\begin{lemma}
	\label{lem:kgfuncpolar}
	Let $F=\sum_{n=0}^Nt_n\in\J(\Ms)$ for some spacetime $\Ms$, where $t_n\in\F^n(\Ms)$ for each $n=0,1,\ldots,N$. Then
	\[
		\tenpow{\EM} nt_n=0
	\]
	as a (nonlinear) functional on $\Ccs(\Ms)$ for all $n$.
\end{lemma}
\begin{proof}
	If $F=\sum_{n=0}^N t_n\in\J(\Ms)$ with $t_n\in\F^n(\Ms)$, then for any $f\in\Ccs(\Ms)$ and $\kappa\in\mathbb{R}$ we have
\[
	0=F[\EM(\kappa f)]=\sum_{n=0}^N\kappa^nt_n[\EM f].
\]
Consequently $t_n[\EM f]=0$ for each $n$, and so by \eqref{eqn:funcform}, and using the fact that for any $g,g'\in\Ccs(\Ms)$ we have $\int_\Ms dx \,g(x)\EM g'(x)=-\int_\Ms dx\,g'(x)\EM g(x)$, it follows that
\[
	(\tenpow\EM nt_n)[f]=(-1)^nt_n[\EM f]=0.
\]
This holds for all $f\in\Ccs(\Ms)$, so the result follows.
\end{proof}

The ideal $\J(\Ms)$ generates an equivalence relation $\sM$; i.e.\ for any $F,F'\in\F(\Ms)$, $F\sM F'$ if and only if $F-F'\in\J(\Ms)$, or equivalently $F[\EM f]=F'[\EM f]$ for all $f\in\Ccs(\Ms)$. For any $F\in\F(\Ms)$, the equivalence class of $F$ under $\sM$ is denoted $[F]_\Ms$; the elements of the algebra $\A(\Ms)$ constitute the set of equivalence classes $[F]_\Ms$ with $F\in\F(\Ms)$.

Throughout this paper, we will wish to define the pullback of a \Loc-arrow $\psi:\Ns\hookrightarrow\Ms$ on a functional $F=\sum_{n=0}^Nt_n\in\F(\Ms)$, with $t_n\in\F^n(\Ms)$. Therefore, we denote
\[
	\psi^\ast F=\sum_{n=0}^N(\psi^{\otimes n})^\ast t_n.
\]
In order to construct the Klein-Gordon QFT as a locally covariant theory, we must now define the action of the \Alg-arrow $\A(\psi)$ for an arbitrary \Loc-arrow $\psi:\Ns\hookrightarrow\Ms$. Given such a $\psi$, we first define a map
\begin{align*}
	\F(\psi):\F(\Ns)&\to\F(\Ms)\\
	F&\mapsto F\circ\psi^\ast.
\end{align*}
To see that $\F(\psi)F$ is indeed an element of $\F(\Ms)$ for any $F\in\F(\Ns)$, note that for $F=\sum_{n=0}^Nt_n,$ $t_n\in\F^n(\Ns)$, we have
\begin{align*}
	(\F(\psi)F)[f]&=\sum_{n=0}^N\int_{\Ns^{\times n}}d^nx\,t_n(x_1,\ldots,x_n)f(\psi(x_1))\cdots f(\psi(x_n))\\
	&=\sum_{n=0}^N\int_{\Ms^{\times n}}d^nx\,\psi_\ast t_n(x_1,\ldots,x_n)f(x_1)\cdots f(x_n)
\end{align*}
for any $f\in\Cs(\Ms)$, where the pushforward $\psi_\ast:\F^n(\Ns)\to\F^n(\Ms)$ is defined as
\[
	\psi_\ast t_n(x_1,\ldots,x_n)=\begin{cases}t_n(\psi^{-1}(x_1),\ldots,\psi^{-1}(x_n)),&(x_1,\ldots,x_n)\in\psi(\Ns)^{\times n}\\
	0,&\text{otherwise}.
	\end{cases}
\]
Since $\psi^{-1}:\psi(\Ns)\to\Ns$ is a diffeomorphism, it follows that each $\psi_\ast t_n$ is an element of $\F^n(\Ms)$ as required. We define the action of $\psi_\ast$ on arbitrary $F\in\F(\Ns)$  by linearity, and note that $\psi^\ast\psi_\ast F=F$. For $F\in\F(\Ms)$, it also holds that $\psi_\ast\psi^\ast F=F$ if and only if the $n^\text{th}$ component of $F$ is supported in $\psi(\Ns)^{\times n}$ for $1\leq n\leq O(F)$. We may naturally define the push-forward on elements of $\Ccs(\Ns)$ by identifying it with the push-forward on $\F^1(\Ns)$.

We will now construct the map $\A(\psi):\A(\Ns)\to\A(\Ms)$ for a \Loc-arrow $\psi:\Ns\hookrightarrow\Ms$, and demonstrate that under this definition $\A$ becomes a covariant functor from \Loc\ to \Alg. 
\begin{lemma}
	\label{lem:kgfmor}
	Let $\Ns,\Ms$ be 
	objects in \Loc, and $\psi:\Ns\hookrightarrow\Ms$ be a \Loc-arrow. Then, for any $F,F'\in\F(\Ns)$ we have
	$F\sim_\Ns F'$ if and only if $\F(\psi)F\sM\F(\psi)F'$.
\end{lemma}
\begin{proof}
	If $\F(\psi)F\sim_\Ms \F(\psi)F'$ then we have $(\F(\psi)F)[\EM g]=(\F(\psi)F')[\EM g]$ for every $g\in\Ccs(\Ms)$. Now, for every $f\in\Ccs(\Ns)$ it holds that $E_\Ns f=\psi^\ast\EM\psi_\ast f$; since $\psi_\ast f\in\Ccs(\Ms)$, it follows that 
\[
F[E_\Ns f]=(\F(\psi)F)[\EM\psi_\ast f]=(\F(\psi)F')[\EM\psi_\ast f]=F'[E_\Ns f].
\]
Therefore $F\sim_\Ns F'$.
	
	  Now suppose that $F\sim_\Ns F'$. Since $O(F),$ $O(F')$ are finite, it follows that there is a compact region $K\subset\Ns$ with the property that the support of the $n^\text{th}$ components of both $F$ and $F'$ lie within $K^{\times n}$ for $1\leq n\leq\max(O(F),O(F'))$. Let $\Sigma_\Ns$ be a Cauchy surface for $\Ns$, and consider the intersection $S=J_\Ns(K)\cap\Sigma_\Ns$; for any classical solution $\EM f,$ $f\in\Ccs(\Ms)$, it will always be possible to pick a smooth pair of functions $(\varphi_f,\pi_f)$ on $\Sigma_\Ns$ which are compactly supported and coincide with the Cauchy data for $\psi^\ast\EM f$ on $S$ (even if $\psi(\Sigma_\Ns)$ cannot be extended to a Cauchy surface for $\Ms$). But since $(\varphi_f,\pi_f)$ are compactly supported they provide data for a solution $E_\Ns g$, for some $g\in\Ccs(\Ns)$. It then holds that $E_\Ns g$ must coincide with $\psi^\ast\EM f$ on the domain of determinacy of $S$; since this region contains $K$, it holds that $(E_\Ns g)|_K=(\psi^\ast\EM f)|_K$. It follows that
\begin{align*}
	(\F(\psi)F)[\EM f]&=(\F(\psi)F)[\psi_\ast E_\Ns g]=F[E_\Ns g]\\
	&=F'[E_\Ns g]=(\F(\psi)F')[\psi_\ast E_\Ns g]=(\F(\psi)F')[\EM f].
\end{align*}	  
Since the choice of $f\in\Ccs(\Ms)$ was arbitrary, we may conclude that $\F(\psi)F\sim_\Ms\F(\psi)F'$.
\end{proof}
\begin{lemma}
	\label{lem:fstarmon}
	Let $\Ns,\Ms$ be 
	objects in \Loc, and $\psi:\Ns\hookrightarrow\Ms$ be a \Loc-arrow. Then $\F(\psi)$ is a $\ast$-monomorphism.
\end{lemma}
\begin{proof}
	Let $F\in\F(\Ns)$ and $f\in\Cs(\Ms)$. Writing $F=\sum_{n=0}^Nt_n$ with $t_n\in\F^n(\Ns)$, we have
\[
	\F(\psi)(F^\ast)=\sum_{n=0}^N\psi_\ast \bar t_n=\sum_{n=0}^N\bar{\psi_\ast t_n}=(\F(\psi)F)^\ast.
\]
Now let $F,F'\in\F(\Ns)$; we have
\begin{align*}
	\F(\psi)(F\star F')&=\sum_k\frac{i^k}{2^kk!}\F(\psi)
	\left[E_\Ns^k\left(F^{(k)}[\,\cdot\,],F'^{(k)}[\,\cdot\,]\right)\right]\\
	&=\sum_k\frac{i^k}{2^kk!}E_\Ns^k\left(F^{(k)}[\psi^\ast\,\cdot\,],F'^{(k)}[\psi^\ast\,\cdot\,]\right).
\end{align*}
But for any distributions $t,t'\in \F^k(\Ns)$, we have
\[
	E_\Ns^k(t,t')=\EM^k(\psi_\ast t,\psi_\ast t'),
\]
and it is also easy to see that for any $F\in\F(\Ns)$, we have $F^{(k)}[\psi^\ast f]=\psi_\ast(\F(\psi)F)^{(k)}[f]$. It follows that
\begin{align*}
	\F(\psi)(F\star F')&=\sum_k\frac{i^k}{2^kk!}\EM^k\left((\F(\psi)F)^{(k)}[\,\cdot\,],
	(\F(\psi)F')^{(k)}[\,\cdot\,]\right)\\
	&=(\F(\psi)F)\star(\F(\psi)F').
\end{align*}
It remains to show that $\F(\psi)$ is injective. If $F,F'\in\F(\Ns)$ with $F\neq F'$, there exists some $f\in\Ccs(\Ns)$ with $F[f]\neq F'[f]$; it follows that $(\F(\psi)F)[\psi_\ast f]\neq(\F(\psi)F')[\psi_\ast f]$, and therefore $\F(\psi)F\neq\F(\psi)F'$.
\end{proof}
The final result to prove for $\F(\psi)$ is that it is indeed a covariant functor:
\begin{lemma}
	\label{lem:fisfunctor}
	The map $\F:\Loc\to\Alg$ which maps an object $\Ms$ to $\F(\Ms)$ and an arrow $\psi$ to $\F(\psi)$ is a covariant functor.
\end{lemma}
\begin{proof}
	Lemma \ref{lem:fstarmon} shows that for a \Loc-arrow $\psi:\Ns\hookrightarrow\Ms$, the map $\F(\psi)$ is indeed an arrow from $\F(\Ns)$ to $\F(\Ms)$. All that remains to prove is that $\F(\text{id}_\Ms)=\text{id}_{\F(\Ms)}$ for any spacetime $\Ms$, and that $\F(\psi_2)\circ\F(\psi_1)=\F(\psi_2\circ\psi_1)$ for any composable \Loc-arrows $\psi_1,\psi_2$. These result directly from the observations that $\text{id}_\Ms^\ast f=f$ for any $f\in\Cs(\Ms)$, and that $\psi_1^\ast\circ\psi_2^\ast=(\psi_2\circ\psi_1)^\ast$.
\end{proof}

We now define the map
\begin{align*}
	\A(\psi):\A(\Ns)&\to\A(\Ms)\\
	[F]_\Ns&\mapsto[\F(\psi)F]_\Ms.
\end{align*}
We can see from lemma \ref{lem:kgfmor} that this map is well defined, and indeed injective; it must also be a $\ast$-homomorphism, as a direct result of the properties of $\F$ proved in lemma \ref{lem:fstarmon}. We can therefore prove the following:
\begin{corollary}
	\label{cor:aisfunctor}
	The map $\A:\Loc\to\Alg$ which maps an object $\Ms$ to $\A(\Ms)$ and an arrow $\psi$ to $\A(\psi)$ is a covariant functor.
\end{corollary}
\begin{proof}
	We have already shown that for any \Loc-arrow $\psi:\Ns\hookrightarrow\Ms$, the map $\A(\psi)$ is an \Alg-arrow from $\A(\Ns)$ to $\A(\Ms)$. The required properties for $\A$ to be a covariant functor follow directly from lemma \ref{lem:fisfunctor}.
\end{proof}

\begin{lemma}
	\label{lem:aalgimage}
	Let $\psi:\Ns\hookrightarrow\Ms$ be an arrow in \Loc. Then $A\in\A(\psi)(\A(\Ns))$ if and only if there exists $F\in\F(\Ms)$ such that $A=[F]_\Ms$, and $F[\EM f]=F[0]$ for every $f\in\Ccs(\Ms)$ such that $\supp(f)\cap J_\Ms(\Ns)=\emptyset$. Moreover, the theory $\A$ is causal.
\end{lemma}
\begin{proof}
Note that $\A(\psi)(\A(\Ns))$ comprises those elements $A\in\A(\Ms)$ that can be represented by those $F\in\F(\Ms)$ with $F=\sum_{n=0}^Nt_n,$ $t_n\in\F^n(\Ms)$, with the property that each $t_n$ can be written as $\psi_\ast t_n'$ for some $t_n'\in\F^n(\Ns)$. But these are precisely those $F=\sum_{n=0}^Nt_n$ for which $\supp(t_n)\subseteq\Ns^{\times n}$ for $n\geq1$, and so for such an $F$ we have $F[f]=F[0]$ for all $f\in\Ccs(\Ms)$ with $\supp(f)\cap\Ns=\emptyset$. Since $F\sim_\Ms F'$ if and only if $F[\EM f]=F'[\EM f]$ for all $f\in\Ccs(\Ms)$, it follows that $F$ represents an element of $\A(\psi)(\A(\Ns))$ if and only if $F[\EM f]=F[0]$ for all $f\in\Ccs(\Ms)$ with $\supp(f)\cap J_\Ms(\Ns)=\emptyset$.

Now suppose that $\Ns_1$ and $\Ns_2$ are spacetimes embedded in $\Ms$ by \Loc-arrows $\psi_1,\psi_2$ respectively, and that $\psi_1(\Ns_1)$ and $\psi_2(\Ns_2)$ are causally disjoint in $\Ms$. It follows that if $A_i\in\A(\psi_i)(\A(\Ns_i))$, $i=1,2$, we may pick $F_1,F_2\in\F(\Ms)$ such that $[F_i]_\Ms=A_i$, and that the $n^\text{th}$ component of $F_i$ is supported in $(\Ns_i)^{\times n}$. It is then clear from \eqref{eqn:kgproddef},\eqref{eqn:kgproddef2} that $(F\star F')[f]=F[f]F'[f]$ for any $f\in\Ccs(\Ms)$. It follows that $[A_1,A_2]=0$, and therefore the theory is causal.
\end{proof}

As a final note on this construction, we remark that a different construction of the Klein-Gordon scalar field theory is given in \cite{dynloc2}, where dynamical locality is proved in the massive minimally coupled case. The construction given above has the advantage that one is able to easily work with the elements of the algebra $\A(\Ms)$ themselves, rather than its generators only; this makes it easy to compute the relative Cauchy evolution for an arbitrary element directly. There is also a natural extension of this construction to the theory of Wick polynomials.

\subsection{Construction of the Theory of Wick Polynomials}
\label{sect:wpconst}

We can extend the construction of the Klein-Gordon theory to a larger theory containing the Wick polynomials. The general aim is to include in the algebras of functionals previously denoted $\F(\Ms)$ a greater range of distributions. The resulting enlarged theory will be denoted $\W$. The following construction follows \cite{bf:qftcb} and \cite{chilfred}. 

We first need to establish the behaviour of the fundamental solution $\EM$ and the Klein-Gordon operator $\PM$ on distributions. For a distribution $t\in\mathcal{D}'(\Ms)$ (resp. $\mathcal{E}'(\Ms)$, i.e.\ compactly supported distributions), and arbitrary $f\in\Ccs(\Ms)$ (resp.\ $\Cs(\Ms)$), we simply define
\[
	\left\langle\PM t,f\right\rangle=\left\langle t,\PM f\right\rangle.
\]
Since $\PM$ is a formally self-adjoint linear differential operator, the restriction of the map $\PM:\mathcal{D}'(\Ms)\to\mathcal{D}'(\Ms)$ to $\Cs(\Ms)$ is compatible with the previous definition of $\PM$ on smooth functions.

Now, analogously to the case for smooth functions, we now wish to construct maps $\bar{\EM^\pm}:\mathcal{E}'\to\mathcal{D}'$ satisfying
\begin{align}
	\bar{\EM^\pm}\PM t&=t=\PM\bar{\EM^\pm} t\label{eqn:Eprop1dist}\\
	\supp(\bar{\EM^\pm} t)&\subseteq J_\Ms^\pm(\supp(t)).\label{eqn:Eprop2dist}
\end{align}
We therefore let $\bar{\EM^\pm} t=(\EM^\mp)'t$: this expression is clearly a well-defined element of $\mathcal{D}'(\Ms)$ for any $t\in\mathcal{E}'(\Ms)$, and this definition ensures that \eqref{eqn:Eprop1dist} is satisfied. Moreover, we may see that \eqref{eqn:Eprop2dist} is satisfied by noting that for any $t\in\mathcal{E}'(\Ms)$, $f\in\Ccs(\Ms)$, we have $J_\Ms^\pm(\supp(t))\cap\supp(f)=\emptyset$ if and only if $\supp(t)\cap J^\mp_\Ms(\supp(f))=\emptyset$. We know that the maps $\EM^\pm$ satisfying \eqref{eqn:funsol1}, \eqref{eqn:funsol2} are unique, so the restrictions of $\bar{\EM^\pm}$ to $\Ccs(\Ms)$ must coincide with $\EM^\pm$.
As before, we let the fundamental solution $\EM:\mathcal{E}'(\Ms)\to\mathcal{D}'(\Ms)$ be defined by $\bar{\EM}=\bar{\EM^-}-\bar{\EM^+}$, and therefore $\bar{\EM}=-(\EM)'$, as would be expected from the relation \eqref{eqn:Eantiprop}.
 From now on, we will drop the bar from the notation and simply write $\EM^{(\pm)} t$ for a distribution $t\in\mathcal{E}'(\Ms)$.

Recall that for any spacetime $\Ms$, the algebra of functionals $\F(\Ms)$ consists of elements of the form $F=\sum_{n=0}^Nt_n$, with each $t_n\in\F^n(\Ms)$ being a finite sum of finite products of test functions of one variable. We wish to include a much wider range of allowed distributions into the new theory $\W$, but we must apply enough restrictions to ensure that the resulting expressions are well defined. We might na\"{i}vely assume that we can use the same product as defined in \eqref{eqn:kgproddef} for distributions, but this is not the case. For example, consider two elements $t,t'\in\F^1(\Ms)$; we see that for any $f\in\Cs(\Ms)$,
\[
	(t\star t')[f]=\int_{\Ms^{\times 2}}dx\,dy\,t(x)t'(y)\left(f(x)f(y)+\frac i2\EM(x,y)\right);
\]
again, for $t\in\mathcal{E}'(\Ms^{\times n})$ we use the notation
\[
	t[f]=\left\langle t,\tenpow fn\right\rangle
	=\int_{\Ms^{\times n}}d^{n-1}x\,t(x_1,\ldots,x_n)f(x_1)\cdots f(x_n),
\]
so for any $f\in\Ccs(\Ms)$ we have $t[\EM f]=(-1)^n(\tenpow\EM nt)[f].$
When $t$ and $t'$ are test functions the second term above is well defined, but pointwise products of distributions are not always so, and we require both a condition on the existence of such pointwise products and a deformation of the product to ensure that all the expressions are well defined. We can find a suitable condition for existence of pointwise products in \cite{hormander1}, namely \emph{H\"ormander's criterion}: If $t$ and $t'$ are distributions, then the pointwise product $t(x)t'(x)$ is a well-defined distribution if the set
\[
	\{(x,k+k'):(x,k)\in WF(t),\ (x,k')\in WF(t')\}
\]
contains no element of the form $(x,0)$.

It is well known (see e.g.\ \cite{duistermaathormander}) that the wavefront set of the distribution $\EM(x,y)$ satisfies
\[
	WF(\EM)\subset \bigcup_{\substack{x,y\in\Ms\\x\leftrightarrow y}}
	 (V^+_{\Ms;x}\times V^-_{\Ms;y})\cup (V^-_{\Ms;x}\times V^+_{\Ms;y}),
\]
where $V^\pm_{\Ms;x}\subset T^\ast_x\Ms$ is the forward/backward light cone at $x$, and $x\leftrightarrow y$ indicates that $x$ and $y$ are connected by a null geodesic. We denote by $V^\pm_\Ms$ the union $\bigcup_{x\in\Ms}V^\pm_{\Ms;x}$. We then define for $n\geq1$ (cf.\ \cite{chilfred})
\[
	\T^n(\Ms)=\{t\in\mathcal{E}'(\Ms^{\times n}):t\text{ totally symmetric},\ WF(t)\cap \bar{(V^+_\Ms)^{\times n}\cup (V_\Ms^-)^{\times n}}=\emptyset\}.
\]
As before we also define $\T^0(\Ms)=\mathbb{C}$.
Such a definition ensures that the expression $\int_{\Ms^{\times2}}dx\,dy\,t(x)t'(y)\EM(x,y)$ for $t,t'\in\T^1(\Ms)$ is well defined (and more generally, that
\[
	\int_{\Ms^{\times2}}dx_1\,dy\,t_n(x_1,\ldots,x_n)t'(y)\EM(x_1,y)
\]
for $t_n\in\T^n(\Ms)$, $t\in\T^1(\Ms)$ is always a well defined element of $\T^{n-1}(\Ms)$). Analogously to the previous case, we wish to define an algebra $\T(\Ms)$ comprising elements of the form
\beq
	\label{eqn:thdef}
	T=\sum_{n=0}^Nt_n
\eeq
with $t_n\in\T^n(\Ms)$. For any $f\in\Cs(\Ms)$ and $T$ of the above form we define the functional derivative $\T^{(k)}[f]$ in the same way as detailed in \eqref{eqn:kgfuncderiv1} and \eqref{eqn:kgfuncderiv2}. It is clear from the definition of $\T^k(\Ms)$ that the functional derivative $\T^{(k)}[f]$ is an element of $\T^k(\Ms)$.

It is shown in \cite{chilfred} that for any $t\in\T^n(\Ms)$, the wavefront set of $(\EM^\pm)_k t$ has the property that $WF((\EM^\pm)_k t) \cap\bar{(V_\Ms^+)^{\times n}\cup(V_\Ms^-)^{\times n}}=\emptyset$, where $(\EM^\pm)_k=\tenpow\id{k-1}\otimes\EM^\pm\otimes\tenpow\id{n-k}$. Since differential operators and multiplication by smooth functions cannot enlarge the wavefront set of a distribution, it follows that any element of $\mathcal{E}'(\Ms^{\times n})$ which is obtained via application of any such operators and $(\EM^\pm)_k$ on an element of $\T^n(\Ms)$ must itself be an element of $\T^n(\Ms)$.

Unfortunately, the restriction on elements of $\T^n(\Ms)$ alone does not solve the problem of ill-defined distributions. Note that for any $g\in\Ccs(\Ms)$, the distribution $t_2(x,y)=g(x)\delta(x-y)$ has empty wavefront set, and is therefore an element of $\T^2(\Ms)$; however
\[
	(t_2\star t_2)[f]=\int_{\Ms^{\times 2}}dx\,dy\,t(x)t(y)\left(f(x)^2f(y)^2+2i
	\EM(x,y)f(x)g(y)-\frac12\EM(x,y)^2\right),
\]
and the distribution $\EM(x,y)^2$ is ill-defined since it does not obey H\"ormander's criterion. A solution to this problem is given in \cite{bf:qftcb}: on each spacetime $\Ms$, it is possible to find symmetric distributions $H$ which satisfy the properties
\[
	WF(\EM+2iH)=WF(\EM)\cap(V^+_\Ms\times V^-_\Ms)
\]
and
\beq
	\label{eqn:hbisol}
	H(\PM f,f')=0
\eeq
for all $f,f'\in\Ccs(\Ms)$.
There is no unique choice for $H$, and we denote by $\H(\Ms)$ the set of all such distributions. It follows that the distribution $(\EM+2iH)^k$ is well defined for any $k\geq1$ and $H\in\H(\Ms)$, and consequently we define a new product $\star_H$ that acts on distributions as
\[
	(T\star_H T')[f]=\sum_{k=0}^{\text{min}(O(T),O(T'))}\frac{i^k}{2^kk!}E_{\Ms;H}^k\left(F^{(k)}[f],F'^{(k)}[f]\right),
\]
where for $t,t'\in\T^k(\Ms)$, we define
\[
	E_{\Ms;H}^k(t,t')=\int_{\Ms^{\times(2k)}}d^kx\,d^ky\,t(x_1,\ldots,x_k)t'(y_1,\ldots,y_k)\prod_{j=1}^k 
	(E_\Ms(x_j,y_j)+2iH(x_j,y_j))
\]
for $k\geq1$. As before, we define $E_{\Ms;H}^0(\alpha,\beta)=\alpha\beta$. One can show that this product is still associative. We then denote by $\T_H(\Ms)$ the algebra comprising elements of the form given in \eqref{eqn:thdef} with product $\star_H$. Addition and involution on $\T_H(\Ms)$ are again given by addition and complex conjugation of distributions respectively.

It is possible to show that for any pair $H,H'\in\H(\Ms)$, the difference $H-H'$ is smooth \cite[Theorem 6]{bf:qftcb}, and also that the algebras $\T_H(\Ms)$ and $\T_{H'}(\Ms)$ are isomorphic; if we define the map
\begin{align}
	\lambda_{H,H'}:\T_H(\Ms)&\to\T_{H'}(\Ms)\notag\\
	T&\mapsto\sum_{n=0}^{\lfloor O(T)/2\rfloor}\frac1{n!}\left\langle 
	\tenpow{(H-H')}n,T^{(2n)}\right\rangle\label{eqn:lambdadef}
\end{align}
where for $t\in\T^{2n}(\Ms)$,
\beq
	\left\langle\tenpow{(H-H')}n,t\right\rangle=\int_{\Ms^{\times(2n)}}d^{2n}x\,t(x_1,\ldots,x_{2n})\prod_{j=1}^n
	(H(x_{2j-1},x_{2j})-H'(x_{2j-1},x_{2j})),
	\label{eqn:anglefuncdef}
\eeq
then this is an isomorphism satisfying $\lambda_{H,H'}=\lambda_{H',H}^{-1}$, $\lambda_{H',H''}\circ\lambda_{H,H'}=\lambda_{H,H''}$ and
\[
	T\star_{H}T'=\lambda_{H,H'}^{-1}(\lambda_{H,H'}(T)\star_{H'}\lambda_{H,H'}(T')).
\]

In exactly the same way that the set $\J(\Ms)$ is an ideal for $\F(\Ms)$, it also holds that the analogous set
\[
	\JJ(\Ms)=\{T\in\T_H(\Ms):T[\EM f]=0\text{ for all }f\in\Ccs(\Ms)\}
\]
(which is independent of the choice of $H\in\H(\Ms)$) is an ideal for $\T_H(\Ms)$. We therefore define the algebra $\W_H(\Ms)=\T_H(\Ms)/\JJ(\Ms).$ Since the equivalence class of an element $T\in\T_H(\Ms)$ does not depend on $H$, we will denote it unambiguously by $[T]_\Ms$, and if $T-T'\in\JJ(\Ms)$ we will write $T\sim_\Ms T'$ as before. It follows from \eqref{eqn:hbisol} that $T\sim_\Ms T'$ if and only if $\lambda_{H,H'}T\sim_\Ms\lambda_{H,H'}T'$, so the isomorphism
\begin{align*}
	\tl_{H,H'}:\W_H(\Ms)&\to\W_{H'}(\Ms)\\
	[T]_{\Ms}&\mapsto[\lambda_{H,H'}T]_{\Ms}
\end{align*}
is well defined. We also note that the reasoning used to show lemma \ref{lem:kgfuncpolar} can be similarly used to show the corresponding result; that if $T\in\T_H(\Ms)$ can be written $T=\sum_{n=0}^Nt_n$ with $t_n\in\T^n(\Ms)$ for each $n$, then $T\in\JJ(\Ms)$ if and only if $t_0=0$ and
\beq
	\label{eqn:wpfuncpolar}
	\tenpow{\EM} n t_n=0
\eeq
 for all $n=1,\ldots,N$.

Since there is no preferred method of uniquely specifying some $H\in\H(\Ms)$ for each spacetime $\Ms$, the above construction does not constitute a locally covariant theory, as we have not yet defined a unique algebra for each $\Ms$. We therefore wish to construct an algebra $\W(\Ms)$ which is independent of the choice of $H$. Again following \cite{bf:qftcb}, we do this by letting $\W(\Ms)$ comprise families of elements indexed by choice of $H\in\H(\Ms)$, as follows:
\[
	\W(\Ms)=\{(W_H)_{H\in\H(\Ms)}:\tl_{H,H'}W_H=W_{H'}\text{ for all }H,H'\in\H(\Ms)\}.
\]
Given $W=(W_H)_{H\in\H(\Ms)}$, $W'=(W'_H)_{H\in\H(\Ms)}$, we define $(W+W')_H=W_H+W'_H$, $(W\star W')_H=W_H\star_H W'_H$ and $(W^\ast)_H=W_H^\ast$. These operations are clearly consistent with the compatibility condition $\tl_{H,H'}W_H=W_{H'}$. Since this condition also ensures that each family $W=(W_H)_{H\in\H(\Ms)}\in\W(\Ms)$ is completely defined by any single entry $W_H$, it follows that $\W(\Ms)\cong\W_H(\Ms)$ for any $H\in\H(\Ms)$.

Having given a prescription for defining $\W(\Ms)$, we must now find a suitable definition for the \Alg-arrow $\W(\psi)$ corresponding to a \Loc-arrow $\psi:\Ns\hookrightarrow\Ms.$ Throughout this section we will use the same notation as before for the definition of the pullback and pushforward of a \Loc-arrow on an arbitrary functional.
\begin{lemma}
Let $\Ns$, $\Ms$ be locally covariant theories, and let $\psi:\Ns\hookrightarrow\Ms$ be an arrow in \Loc. Then for any $H\in\H(\Ms),$ we have $\psi^\ast H\in\H(\Ns)$.
\end{lemma}
\begin{proof}
	We have $WF(\phi^\ast t)\subseteq \phi^\ast WF(t)$ for any smooth $\phi:\Ns\to\Ms$ and distribution $t$ on $\Ms$ \cite[Theorem 2.5.11$'$]{hormanderfio1}. It is a clear consequence that we have equality whenever $\phi$ is a local diffeomorphism; this entails that when $\psi:\Ns\hookrightarrow\Ms$ is an arrow in \Loc, we have $WF(\psi^\ast T)=\psi^\ast WF(T)$ for any $T\in\mathcal{D}'(\Ms^{\times n})$. 
 Therefore
\[
	WF(E_\Ns+2i\psi^\ast H)=\psi^\ast WF(\EM+2iH)=WF(E_\Ns)\cap(V_\Ns^+\times V_\Ns^-).
\]
Moreover, if $H(\PM f,f')=0$ for all $f,f'\in\Ccs(\Ms)$, it follows that $\psi^\ast H(P_\Ns f,f')=H(\PM\psi_\ast f,\psi_\ast f')=0$ for all $f,f'\in\Ccs(\Ns)$. Therefore $\psi^\ast H\in\H(\Ns)$.
\end{proof}
Note that for any \Loc-arrow $\psi:\Ns\hookrightarrow\Ms$, we may also say that $WF(\psi_\ast U)=\psi_\ast WF(U)$ for $U\in\mathcal{E}'(\Ns^{\times n})$.\footnote{We require compact support of $U$ here; if $U\in\mathcal{D}'(\Ns)$, then we might not have equality, although $(x_1,\ldots,x_n;k_1,\ldots,k_n)\in WF(\psi_\ast U)\setminus\psi_\ast WF(U)$ only if $x_k\in\partial(\psi(\Ns))$ for each $k$.}

Now, for any $H\in\H(\Ms)$ we define the map
\begin{align*}
	\T_H(\psi):\T_{\psi^\ast H}(\Ns)&\to\T_H(\Ms)\\
	T&\mapsto T\circ\psi^\ast.
\end{align*}
For any $T=\sum_{n=0}^Nt_n\in\T_{\psi^\ast H}(\Ns)$, $t_n\in\T^n(\Ns)$, we have 
\[
	(\T_H(\psi)T)[f]=\sum_{n=0}^N\psi_\ast t_n[f]
\]
as before, and since $t_n$ is compactly supported for each $n\geq1$, it follows that $WF(\psi_\ast t_n)= \psi_\ast WF(t_n)$. Thus $\T_H(\psi)T$ is an element of $\T_H(\Ms)$ as required. We can also use the same argument as for lemma \ref{lem:kgfmor} to see that for any $T,T'\in\T_{\psi^\ast H}(\Ns)$, it holds that $T\sim_\Ns T'$ if and only if $\T_H(\psi)T\sim_\Ms\T_H(\psi)T'$. Moreover, the result of lemma \ref{lem:fstarmon} extends directly to $\T_H(\psi)$, so it is indeed a $\ast$-monomorphism. Therefore the map
\begin{align}
	\W_H(\psi):\W_{\psi^\ast H}(\Ns)&\to\W_H(\Ms)\notag\\
	[T]_\Ns&\mapsto[\T_H(\psi)T]_\Ms\label{eqn:WHpsidef}
\end{align}
is a well-defined $\ast$-monomorphism. From this, we define the map	$\W(\psi):\W(\Ns)\to\W(\Ms)$ by 
\beq
	(\W(\psi)W)_H=\W_H(\psi)W_{\psi^\ast H},
	\label{eqn:Wpsidef}
\eeq
where $H\in\H(\Ms)$.
It is easy to show that this definition is consistent with the compatibility condition: i.e.\ $\tl_{H,H'}(\W(\psi)W)_H=(\W(\psi)W)_{H'}$ for all $H,H'\in\H(\Ms)$. We then have:
\begin{lemma}
	The map $\W:\Loc\to\Alg$ which maps spacetimes $\Ms$ to $\W(\Ms)$ and \Loc-arrows $\psi$ to $\W(\psi)$ is a covariant functor.
\end{lemma}
\begin{proof}
	It is trivial to show that for any spacetime $\Ms$, we have $\W(\text{id}_\Ms)=\text{id}_{\W(\Ms)}$. It remains to show that for any \Loc-arrows $\psi_1:\Ns_1\hookrightarrow\Ns_2,$ $\psi_2:\Ns_2\hookrightarrow\Ms$, it holds that $\W(\psi_2)\circ\W(\psi_1)=\W(\psi_2\circ\psi_1)$. For any $T\in\T_{\psi_1^\ast\psi_2^\ast H}(\Ns_1)$ and $\H\in\H(\Ms)$, we have
\[
	\T_H(\psi_2)\T_{\psi^\ast_2H}(\psi_1)T=T\circ\psi_1^\ast\circ \psi_2^\ast=
	\T_H(\psi_2\circ\psi_1)T.
\]
The desired result follows by \eqref{eqn:WHpsidef}, \eqref{eqn:Wpsidef}.
\end{proof}
The covariant functor $\W$ is thus a locally covariant theory which represents the extended algebra of Wick polynomials. We also have the corresponding result to lemma \ref{lem:aalgimage}:
\begin{lemma}
	\label{lem:walgimage}
	Let $\psi:\Ns\hookrightarrow\Ms$ be an arrow in \Loc. Then $W\in\W(\psi)(\W(\Ns))$ if and only if there exists $T\in\T_H(\Ms)$ such that $W_H=[T]_\Ms$ for some $H\in\H(\Ms)$, and $T[\EM f]=T[0]$ for every $f\in\Ccs(\Ms)$ such that $\supp(f)\cap J_\Ms(\Ns)=\emptyset$. Moreover, the theory $\W$ is causal.
\end{lemma}
\begin{proof}
	$W\in\W(\psi)(\W(\Ms))$ if and only if we have $W_H\in\W_H(\psi)(\W_{\psi^\ast H}(\Ns))$ for some (and consequently every) $H\in\H(\Ms)$; the required results then follow using an analogous argument to that given in the proof of lemma \ref{lem:aalgimage}. 
\end{proof}

\subsection{Spaces of smooth functions on spacetimes}
Before we consider the timeslice axiom and dynamical locality of the two theories, we discuss the following spaces of smooth functions on $\Ms$, in addition to $\Ccs(\Ms)$ and $\Cs(\Ms)$. We define
\begin{align*}
	\Cs_s(\Ms)&=\{f\in\Cs(\Ms):\supp(f)\subseteq J_\Ms(K)
	\text{ for some compact }K\subset\Ms\},\\
	\Cs_{s,\pm}(\Ms)&=\{f\in\Cs_s(\Ms):\supp(f)\subseteq J_\Ms^\pm(K)
	\text{ for some compact }K\subset\Ms\}.
\end{align*}
We also use the following notation for the canonical embeddings
\begin{align*}
	\iota_{0,\pm}&:\Ccs(\Ms)\hookrightarrow\Cs_{s,\pm}(\Ms),\\
	\iota_{\pm,s}&:\Cs_{s,\pm}(\Ms)\hookrightarrow\Cs_s(\Ms),\\
	\iota_{s,\infty}&:\Cs_s(\Ms)\hookrightarrow\Cs(\Ms).
\end{align*}
We wish to demonstrate that there exist continuous maps $\hEM^\pm:\Ccs(\Ms)\to\Cs_{s,\pm}(\Ms)$ that satisfy $\EM^\pm=\iota_{s,\infty}\circ\iota_{\pm,s}\circ\hEM^\pm$. It is clear that for any $f\in\Ccs(\Ms)$, the function $\EM^\pm f$ lies within the range of $\iota_{s,\infty}\circ\iota_{\pm,s}$, we may unambiguously let $\hEM^\pm=(\iota_{s,\infty}\circ\iota_{\pm,\infty})^{-1}\circ\EM^\pm.$ To establish continuity we must first define the topologies on each of these spaces of functions. The spaces $\Cs(\Ms)$ and $\Ccs(\Ms)$ can be constructed as convex topological spaces, as follows \cite{waldmann,reedsimon}. A \emph{compact exhausting sequence} for $\Ms$ is a sequence $(K_n)_{n\in\mathbb{N}}$ of compact submanifolds of $\Ms$ such that $K_n\subset\mathring{K}_{n+1}$ for each $n$, and for every point $p\in\Ms$ there exists $N\in\mathbb{N}$ such that $p\in K_n$ for all $n>N$. Any space of smooth functions on a smooth manifold can be endowed with the $\Cs$ topology; we do not need to go into details here, except to say that the topology on $C^\infty(\Ms)$ is generated by seminorms $p_{K_n,k}$, $k,n\in\mathbb{N}$, where $(K_n)_{n\in\mathbb{N}}$ is a compact exhausting sequence for $\Ms$, and $p_{K_n,k}(f)$ is given by the supremum over $K_n$ of the norms of all covariant derivatives of $f$ of order no greater than $k$ (using a Riemannian metric to induce the norms of the derivatives). The $\Cs$ topology on a space of smooth functions on $\Ms$ is then defined as the subspace topology induced from $\Cs(\Ms)$. The topology of $\Ccs(\Ms)$, on the other hand, is constructed as an inductive limit of the topological spaces $\Cs_{K_n}$ (that is, the finest topology such that each embedding $\iota_n:\Cs_{K_n}(\Ms)\hookrightarrow\Ccs(\Ms)$ is continuous), where $(K_n)_{n\in\mathbb{N}}$ is again a compact exhausting sequence for $\Ms$, and $\Cs_K(\Ms)$ is the space $\{f\in\Cs(\Ms):\supp(f)\subseteq K\}$ endowed with the $\Cs$ topology. Now, for any inductive limit $X$ of locally convex spaces $(X_n)_{n\in\mathbb{N}}$, and locally convex space $Y$, a map $T:X\to Y$ is continuous if and only if each restriction $T|_{X_n}:X_N\to Y$ is continuous \cite[Theorem V.16]{reedsimon}. Since the space $\Cs_{K_n}(\Ms)$ inherits the subspace topology induced from $\Cs(\Ms)$, it follows that the embedding $\Ccs(\Ms)\hookrightarrow\Cs(\Ms)$ is continuous.

Now, for a spacetime $\Ms$ we wish to endow $\Cs_s(\Ms)$ and $\Cs_{s,\pm}(\Ms)$ with  topologies in a similar way to that given for $\Ccs(\Ms)$ in \cite{waldmann,reedsimon}; starting with a compact exhausting sequence $(K_n)_{n\in\mathbb{N}}$ for $\Ms$, we consider the topological spaces $\Cs_{J_\Ms(K_n)}(\Ms)$ and $\Cs_{J_\Ms^\pm(K_n)}(\Ms)$ defined analogously to $\Cs_{K_n}(\Ms)$, and let $\Cs_s(\Ms)$ and $\Cs_{s,\pm}(\Ms)$ be the inductive limit of $\Cs_{J_\Ms(K_n)}(\Ms)$ and $\Cs_{J_\Ms^\pm(K_n)}(\Ms)$ respectively as $n\to\infty$. We then have:
\begin{lemma}
	The embeddings $\iota_{0,\pm}$, $\iota_{\pm,s}$ and $\iota_{s,\infty}$ are all continuous in the relevant topologies.
\end{lemma}
\begin{proof}
	For the sake of readable notation, we denote $X_n=\Cs_{K_n}(\Ms)$, $Y_n^\pm=\Cs_{J_\Ms^\pm(K_n)}(\Ms)$ and $Z_n=\Cs_{J_\Ms(K_n)}$. Firstly, we consider $\iota_{s,\infty}$: for any $n\in\mathbb{N}$, the space $Z_n$ is endowed with the subspace topology induced from $\Cs(\Ms)$, so the embedding must be continuous; therefore $\iota_{s,\infty}|_{Z_n}:Z_n\to\Cs(\Ms)$ is continuous for all $n$, as required for continuity of $\iota_{s,\infty}$. Now, for each $n$ we may factorise $\iota_{\pm,s}|_{Y_n}$ as the composition of the embeddings of  $Y_n^\pm\hookrightarrow Z_n$ and $Z_n\hookrightarrow\Cs_s(\Ms)$; the former is continuous as $Y_n^\pm$ has the subspace topology induced from $Z_n$, and the latter is continuous by definition of $\Cs_s(\Ms)$. Therefore $\iota_{\pm,s}$ is continuous. Similarly, we may factorise $\iota_{0,\pm}|_{X_n}$ as the composition of the embeddings of $X_n\hookrightarrow Y_n^\pm$ and $Y_n^\pm\hookrightarrow\Cs_{s,\pm}(\Ms)$, both of which are continuous. Therefore $\iota_{0,\pm}$ is continuous.
\end{proof}
This also allows us to prove:
\begin{lemma}
	The maps $\hEM^\pm:\Ccs(\Ms)\to\Cs_{s,\pm}(\Ms)$ are continuous.
\end{lemma}
\begin{proof}
	We recall that if a topological space $Y$ is endowed with the subspace topology from a space $Z$, and the embedding is denoted $\iota:Y\hookrightarrow Z$, then a map $T:X\to Y$ is continuous if and only if $\iota\circ T$ is continuous. We note that $X_n=\Cs_{K_n}(\Ms)$ has the subspace topology induced from $\Cs(\Ms)$; since $\EM^\pm:\Ccs(\Ms)\to\Cs(\Ms)$ is continuous, it follows that the restrictions $\EM^\pm|_{X_n}:X_n\to\Cs(\Ms)$ are all continuous. Denoting the canonical embedding by $\iota_n:X_n\hookrightarrow\Cs(\Ms)$, it is clear that we may factorise $\EM^\pm|_{X_n}=\iota_n\circ\hEM^\pm|_{X_n}$, so each $\hEM^\pm|_{X_n}$ is continuous. Therefore $\hEM^\pm$ is continuous.
\end{proof}
We define
\begin{align*}
	\hEM:\Ccs(\Ms)&\to\Cs_s(\Ms)\\
	f&\mapsto\iota_{s,-}(\hEM^-f)-\iota_{s,+}(\hEM^+f),
\end{align*}
which is clearly continuous; we also define $\cEM:(\Cs_s(\Ms))'\to\mathcal{D}'(\Ms)$ by $\cEM=-(\hEM)'$. The map $\PM$ may be considered to act on elements of $\Cs_s(\Ms)$ and $\Cs_{s,\pm}$ in the obvious way, from which we see that strictly speaking $\PM\hEM^\pm f=\iota_{0,\pm}=\hEM^\pm\PM f$ for any $f\in\Ccs(\Ms)$.

  We say that a distribution $t\in\mathcal{D}'(\Ms^{\times n})$ is \emph{time-compact} if there exist spacelike Cauchy surfaces $\Sigma^\pm\subset\Ms$ such that $\supp(t)\subseteq(J_\Ms^-(\Sigma^+)\cap J_\Ms^-(\Sigma^-))^{\times n}.$ Note that the action of a time-compact distribution $t$ is well-defined on $f\in\Cs_s(\Ms^{\times n})$, since the intersection $\supp(t)\cap\supp(f)$ is compact. Therefore any time-compact distribution can be considered to be an element of $(\Cs_s(\Ms^{\times n}))'$. We also say that a distribution $t$ is \emph{future-compact} if there exists a Cauchy surface $\Sigma\subset\Ms$ such that $\supp(t)\subseteq(J_\Ms^-(\Sigma))^{\times n}$, and \emph{past-compact} if there exists a Cauchy surface $\Sigma\subset\Ms$ such that $\supp(t)\subseteq(J_\Ms^+(\Sigma))^{\times n}$. We may similarly see that a future-/past-compact distribution can be considered to be an element of $(\Cs_{s,\pm}(\Ms^{\times n}))'$.

We then have the following result, which will be important later:
\begin{lemma}
	\label{lem:eetapchi}
	Let $u\in\mathcal{D}'(\Ms)$, with $u[\PM f]=0$ for all $f\in\Ccs(\Ms)$. Then there exists a distribution $t\in(\Cs_s(\Ms))'$ such that $u=\cEM t$.
\end{lemma}
\begin{proof}
	Let $\Sigma^\pm$ be two disjoint Cauchy surfaces in $\Ms$ with $\Sigma^+\subset J_\Ms^+(\Sigma^-)$, and let $\cadv+\cret=1$ be a smooth partition of unity with $\cadv(x)=0,$ $\cret(x)=1$ for $x\in J_\Ms^+(\Sigma^+)$ and $\cadv(x)=1,$ $\cret(x)=0$ for $x\in J_\Ms^-(\Sigma^-)$. Now, let $\eta\in\Cs(\Ms)$ be time-compact, and defined such that $\eta(x)=1$ for all $x\in J_\Ms^-(\widetilde\Sigma^+)\cap J_\Ms^+(\widetilde\Sigma^-)$, where $\widetilde\Sigma^\pm\subset\Ms$ are further Cauchy surfaces disjoint from $\Sigma^\pm$ with $\Sigma^\pm\subset J_\Ms^\mp(\widetilde\Sigma^\pm)$. We define the map $\eta_s:\Cs_s(\Ms)\to\Ccs(\Ms)$ by the action of multiplication by $\eta$. We also consider $\chi^{\text{adv/ret}}:\Ccs(\Ms)\to\Ccs(\Ms)$ as defined by action of multiplication. The operator $\PM$ can be considered as an endomorphism acting on any of the spaces of functions we defined above; we may similarly consider it as an endomorphism on any of the dual spaces in question, by 
\[
	\langle \PM u,f\rangle=\langle u,\PM f\rangle.
\]
We may then show that $u=\cEM\eta_s'\PM(\cadv)' u$, as follows. Let $f\in\Ccs(\Ms)$ be arbitrary, then
\begin{align*}
	(\cEM\eta_s'\PM(\cadv)' u)[f]&=-(\eta_s'\PM(\cadv)' u)[\hEM f]\\
	&=(\eta_s'\PM(\cadv)' u)[\iota_{+,s}\hEM^+f]-(\eta_s'\PM(\cadv)' u)
	[\iota_{-,s}\hEM^-f]\\
	&=u[\cadv\PM\eta_s\iota_{+,s}\hEM^+f]-u[\cadv\PM\eta_s\iota_{-,s}
	\hEM^-f].
\end{align*}
Since $u[\PM\eta_s\iota_{-,s}\hEM^-f]=0$, we may use $\cadv=\id-\cret$ to see that
\beq
	(\cEM\eta_s'\PM(\cadv)'u)[f]=u[\cadv\PM\eta_s\iota_{+,s}\hEM^+f]
	+u[\cret\PM\eta_s\iota_{-,s}\hEM^-f].\label{eqn:usplit}
\eeq
Now any $g\in\Cs_{s,+}(\Ms)$ can be split into a sum of three functions $g_-,$ $g_0$ and $g_+$, with the properties that $\supp(g_\pm)\subseteq J_\Ms^\pm(\Sigma^\pm)$ and $\supp(g_0)\subseteq J_\Ms^-(\widetilde\Sigma^+)\cap J_\Ms^+(\widetilde\Sigma^-)$.We may note that $\supp(g_-)$ and $\supp(g_0)$ are both compact, so we can consider $g_0$ and $g_-$ as elements of $\Ccs(\Ms)$, whereupon
\[
	g=\iota_{0,+}g_-+\iota_{0,+}g_0+g_+.
\]
By construction, we have $\eta_s\iota_{+,s}\iota_{0,+}g_0=g_0$; the definition of $\cadv$ also shows that $\cadv T_1g_-=T_1g_-$ and $\cadv T_2g_+=0$ for any operators $T_1:\Ccs(\Ms)\to\Ccs(\Ms)$, $T_2:\Cs_{s,+}(\Ms)\to\Ccs(\Ms)$ such that $\supp(T_i f)\subseteq\supp(f)$ for all $f$, $i=1,2$. It follows directly from these that if we let $g=\hEM^+f$ and split as described, then
\begin{align*}
	\cadv\PM\eta_s\iota_{+,s}\hEM^+f&=
	\cadv\PM\eta_s\iota_{+,s}\iota_{0,+}g_-+
	\cadv\PM\eta_s\iota_{+,s}\iota_{0,+}g_0+
	\cadv\PM\eta_s\iota_{+,s}g_+\\
	&=\PM\eta_s\iota_{+,s}\iota_{0,+}g_-+\cadv\PM g_0.
\end{align*}
But $f=(\iota_{0,+})^{-1}\PM\hEM^+f=\PM g_0+\PM g_-+(\iota_{0,+})^{-1}\PM g_+$, so by the properties of $\cadv$ we have
\begin{align*}
	\cadv\PM\eta_s\iota_{+,s}\hEM^+f&=\PM\eta_s\iota_{+,s}\iota_{0,+}g_-
	+\cadv f-\cadv\PM g_--\cadv(\iota_{0,+})^{-1}\PM g_+\\
	&=\PM(\eta_s\iota_{+,s}\iota_{0,+}-\id)g_-+\cadv f.
\end{align*}
Since $u$ is a weak solution and $(\eta_s\iota_{+,s}\iota_{0,+}-\id)g_-$ is compactly supported, we have
\[
	u[\cadv\PM\eta_s\iota_{+,s}\hEM^+f]=u[\cadv f].
\]
We may similarly conclude that $u[\cret\PM\eta_s\iota_{-,s}\hEM^-f]=u[\cret f].$ It follows from \eqref{eqn:usplit} that
\begin{align*}
	(\cEM\eta_s'\PM(\cadv)'u)[f]&=u[\cadv f]+u[\cret f]\\
	&=u[f].
\end{align*}
This proves the required result, and also gives us an explicit example of a distribution $t\in(\Cs_s(\Ms))'$ satisfying $u=\cEM t$.
\end{proof}
While we have been very careful with our definitions in this section, in the remainder of the paper we will not need to be so exact with our notation. Firstly, we make the observation that since any multiplication operator $\mu$ between spaces of smooth functions is formally self-adjoint, it makes sense to write $\mu't=\mu t$ for a distribution $t$ and formally regard $\mu t$ as the pointwise product of $t$ with the underlying function $\mu\in\Cs(\Ms)$. We will particularly use this convention when a distributional solution $u$ is of the form $u=\EM t$, where $t\in\mathcal{E}'(\Ms)$. Lemma \ref{lem:eetapchi} tells us that $\EM t=\cEM\eta_s'\PM(\cadv)'\EM t=\cEM\eta\PM\cadv\EM t$; however, regarding $\cadv\EM t$ as a pointwise product allows us to see that in fact the distribution $\PM\cadv\EM t$ must be supported within the region $J_\Ms^-(\Sigma^+)\cap J_\Ms^+(\Sigma^-)$ where $\cadv$ is non-constant, by the properties of $\PM$ and $\EM$. Moreover, the support of $\PM\cadv\EM t$ lies within $J_\Ms(\supp(t))$, which has compact intersection with $J_\Ms^-(\Sigma^+)\cap J_\Ms^+(\Sigma^-)$, so the support of $\PM\cadv\EM t$ is compact. Since $\eta=1$ everywhere within $\supp(\PM\cadv\EM t)$, we may suppress $\eta$ and instead regard $\PM\cadv\EM t$ itself as an element of $\mathcal{E}'(\Ms)$, writing
\beq
	\label{eqn:epchidist}
	\EM\PM\cadv\EM t=\EM t.
\eeq
Moreover $\PM\EM t=0$ for any $t\in\mathcal{E}'(\Ms)$, so we also have
\[
	\EM\PM\cret\EM t=-\EM t.
\]

\section{The timeslice axiom and relative Cauchy evolution}
It is well known that both the Klein-Gordon theory \cite{bfv} and the enlarged algebra of Wick polynomials \cite{chilfred} obey the timeslice axiom. However, we will give a proof that the timeslice axiom holds in both cases, since the construction is different to that used in the aforementioned references, and since we require an explicit expression for the inverse maps $\A(\psi)^{-1}$ and $\W(\psi)^{-1}$ when $\psi:\Ns\hookrightarrow\Ms$ is a Cauchy arrow.

\subsection{The timeslice axiom for the Klein-Gordon Theory}
\label{sect:kgtimeslice}
In order to compute the relative Cauchy evolution for either $\A$ or $\W$, we must first demonstrate that they obey the timeslice axiom. It is worth asking first whether the theory $\F$ obeys the timeslice axiom; since the construction for $\F$ contains no condition relating to the field equation, we should not expect $\F$ to obey the axiom, and indeed this is the case: let $\Ns,\Ms$ be objects in \Loc, and $\psi:\Ns\hookrightarrow\Ms$ be a Cauchy arrow in \Loc. Suppose that $\psi(\Ns)\neq\Ms$; then, pick some nonzero $t\in\F^1(\Ms)$ whose support lies within $\Ms\setminus\psi(\Ns)$. Clearly $t\!\left[\bar t\right]\neq 0$, but as $\psi^\ast \bar t=0$, we have $(\F(\psi)F)\!\left[\bar t\right]=0$ for all $F\in\F(\Ns)$. Therefore $\F(\psi)$ is not surjective, and consequently cannot be invertible; hence $\F$ does not satisfy the timeslice axiom.

To demonstrate that $\A$, on the other hand, does obey the timeslice axiom, we use following lemma, which is proved in \cite{dimock} (and can also be seen to be a consequence of lemma \ref{lem:eetapchi}: see \eqref{eqn:epchidist}).
\begin{lemma}
\label{lem:epchi}
Let $\Sadv,\Sret$ be disjoint Cauchy surfaces in a globally hyperbolic spacetime $\Ms$, with $\Sret\subseteq J_\Ms^+(\Sadv)$, and let $\cadv+\cret=1$ be a smooth partition of unity on $\Ms$ with $\cadv(x)=0,\ \cret(x)=1$ for $x\in J_\Ms^+(\Sret)$ and $\cadv(x)=1,\ \cret(x)=0$ for $x\in J_\Ms^-(\Sadv)$. Then 
\begin{align*}
	\EM \PM\cadv \EM f&=\EM f,\\
	\EM \PM\cret \EM f&=-\EM f
\end{align*}
for all $f\in\Ccs(\Ms)$. Moreover, $\PM\chi^\text{adv/ret}\EM f\in\Ccs(\Ms)$.
\end{lemma}

Defining $\zeta t=\PM\cadv \EM t$ for $t\in\F^1(\Ms)$, it follows directly that given $\Sadv,\Sret$ and $\cadv,\cret$ defined as above, for any $t\in\F^n(\Ms)$, $n\geq1$, we have
\[
	\supp(\tenpow{\zeta}nt)\subseteq(J_\Ms^+(\Sadv)\cap J_\Ms^-(\Sret))^{\times n}\cap \supp(\tenpow\EM nt).
\]
Clearly $\tenpow\zeta n$ maps elements of $\F^n(\Ms)$ to elements of $\F^n(\Ms)$. We also note that by lemma \ref{lem:epchi}, we have
\begin{align}
	\tenpow\zeta nt[\EM f]&=(-1)^n\tenpow{(\EM\zeta)}n t[f]\notag\\
	&=(-1)^n\tenpow{\EM} n t[f]\notag\\
	&=t[\EM f]\label{eqn:talphat}
\end{align}
for any $t\in\F^n(\Ms)$, $n\geq1$ and $f\in\Ccs(\Ms)$. It follows that if we define
\begin{align*}
	\Zeta:\F(\Ms)&\to\F(\Ms)\\
	\sum_{n=0}^Nt_n&\mapsto\sum_{n=0}^N\tenpow\zeta nt_n\qquad(t_n\in\F^n(\Ms)),
\end{align*}
we have $\Zeta F\sM F$ for all $F\in\F(\Ms)$.

\begin{lemma}
	\label{lem:atimeslice}
	The theory $\A$ obeys the timeslice axiom.
\end{lemma}
\begin{proof}
	Suppose that $\psi$ is a \Loc-arrow from $\Ns$ to $\Ms$ with the property that $\psi(\Ns)$ contains a Cauchy surface for $\Ms$. We will always be able to find a second Cauchy surface for $\Ms$ in $\psi(\Ns)$ which is disjoint to the first; we denote the Cauchy surface to the past by $\Sadv$ and the one to the future by $\Sret$, and define the operator $\Zeta$ as above using these Cauchy surfaces for the construction; it follows that for any $F\in\F(\Ms)$, the $n^\text{th}$ component of $\Zeta F$ is supported in $\psi(\Ns)^{\times n}$ for each $n\geq1$. We then define 
\begin{align}
	\mathscr{G}(\psi):\F(\Ms)&\to\F(\Ns)\notag\\
	F&\mapsto\psi^\ast\Zeta F.\label{eqn:Gmapdef}
\end{align}
For any $F\in\F(\Ms)$ and $f\in\Cs(\Ms)$, we then have 
\[_{H\in\H(\Ms)}
	(\F(\psi)\mathscr{G}(\psi)F)[f]=(\F(\psi)\psi^\ast\Zeta F)[f]=\psi^\ast\Zeta F[\psi^\ast f]=\psi_\ast\psi^\ast\Zeta F[ f].
\]
But since the $n^\text{th}$ component of $\Zeta F$ is supported in $\psi(\Ns)^{\times n}$, we have $\psi_\ast\psi^\ast\Zeta F=\Zeta F.$ Therefore $\F(\psi)\mathscr{G}(\psi)F=\Zeta F$. Now suppose that $F\in\F(\Ns)$ and $f\in\Cs(\Ns)$; then,
\[
	(\mathscr{G}(\psi)\F(\psi)F)[f]=\mathscr{G}(\psi)(\psi_\ast F)[f]
	=\psi^\ast\Zeta(\psi_\ast F)[f].
\]
Writing $F=\sum_{n=0}^Nt_n$, with $t_n\in\F^n(\Ns)$, we have
\[
	\psi^\ast \Zeta(\psi_\ast F)=\sum_{n=0}^N\psi^\ast\tenpow\zeta n\psi_\ast t_n.
\]
But notice that for any $t\in\F^1(\Ms)$, $f\in\Ccs(\Ms)$, we have
\[
	(\psi^\ast\zeta\psi_\ast t)[E_\Ns f]=(P_\Ns\psi^\ast(\cadv\EM\psi_\ast t))[E_\Ns f]
	=(P_\Ns((\psi^\ast\cadv)E_\Ns t))[E_\Ns f]=t[E_\Ns f]
\]
by \eqref{eqn:talphat} and lemma \ref{lem:epchi}. We have therefore shown that $\F(\psi)\mathscr{G}(\psi)F\sM F$ for all $F\in\F(\Ms)$, and $\mathscr{G}(\psi)\F(\psi)F\sim_\Ns F$ for all $F\in\F(\Ns)$.

Next, we observe that if $F,F'\in\F(\Ms)$ with $F\sM F'$, then we have $\F(\psi)\mathscr{G}(\psi)F\sM\F(\psi)\mathscr{G}(\psi)F'$; by lemma \ref{lem:kgfmor} we then have $\mathscr{G}(\psi)F\sim_\Ns\mathscr{G}(\psi)F'$. This means that the map
\begin{align*}
	\mathscr{B}(\psi):\A(\Ms)&\to\A(\Ns)\\
	[F]_\Ms&\mapsto[\mathscr{G}(\psi)F]_\Ns
\end{align*}
is well defined, and we can conclude that $\mathscr{B}(\psi)\circ\A(\psi)=\text{id}_{\A(\Ns)},$ and $\A(\psi)\circ\mathscr{B}(\psi)=\text{id}_{\A(\Ms)}$. Therefore $\A(\psi)$ is invertible, and so $\A$ obeys the timeslice axiom.
\end{proof}

\subsection{The timeslice axiom for the theory of Wick Polynomials}

\label{sect:wptimeslice}

We now proceed to the timeslice axiom for $\W$, adapting the proof given for an equivalent construction in \cite{chilfred} for the construction used here. 
Suppose that $\psi:\Ns\hookrightarrow\Ms$ is a Cauchy arrow in \Loc. We can then find two disjoint Cauchy surfaces $\Sadv,\Sret\subset\psi(\Ns)$ for $\Ms$ with $\Sret\subset J_\Ms^+(\Sadv)$. As before, we choose a smooth partition of unity $\cadv+\cret=1$ with $\cadv(x)=0,\ \cret(x)=1$ for $x\in J_\Ms^+(\Sret)$ and $\cadv(x)=1,\ \cret(x)=0$ for $x\in J_\Ms^-(\Sadv)$. It follows that if we again define $\zeta t=\PM\cadv\EM t$ for any $t\in\T^1(\Ms)$, and for any $H\in\H(\Ms)$, we let
\begin{align*}
	Z:\T_H(\Ms)&\to\T_H(\Ms)\\
	\sum_{n=0}^Nt_n&\mapsto\sum_{n=0}^N\tenpow\zeta nt_n\qquad (t_n\in\T^n(\Ms)),
\end{align*}
then by \eqref{eqn:epchidist}, $ZT\sim_\Ms T$ for all $T\in\T_H(\Ms)$, and $T$ is compactly supported in $\psi(\Ns)$. Moreover, since $Z$ is constructed from differential operators, multiplication by smooth functions and applications of $\EM^\pm$, we recall from our previous observation that $Z$ must indeed map elements of $\T_H(\Ms)$ to elements of $\T_H(\Ms)$.

Therefore, if we define
\begin{align}
	\mathscr{S}_H(\psi):\T_H(\Ms)&\to\T_{\psi^\ast H}(\Ns)\notag\\
	T&\mapsto \psi^\ast ZT,\label{eqn:tinvdef}
\end{align}
then the same argument as used in the proof of lemma \ref{lem:atimeslice} shows that $\T_H(\psi)\mathscr{S}_H(\psi)T\sim_\Ms T$ for all $T\in\T_H(\Ms)$ and $\mathscr{S}_H(\psi)\T_H(\psi)T\sim_\Ns T$ for all $T\in\T_{\psi^\ast H}(\Ns)$.

Now, if $\psi(\Ns)$ contains a Cauchy surface for $\Ms$ then for each $H\in\H(\Ns)$ there is precisely one $H'\in\H(\Ms)$ with $\psi^\ast H'=H$, as a result of the condition \eqref{eqn:hbisol}. We will denote this extension by $\psi_\bullet H$. Now suppose that $W=(W_H)_{H\in\H(\Ms)}\in\W(\Ms)$ with $W_H=[T_H]$, for some $T_H\in\T_H(\Ms)$ for each $H\in\H(\Ms)$. We then define
\begin{align}
	\mathscr{U}_H(\psi):\W_H(\Ms)&\to\W_{\psi^\ast H}(\Ns)\notag\\
	[T]_\Ms&\mapsto[\mathscr{S}_H(\psi)T]_\Ns.\label{eqn:winvdef}
\end{align}
This then gives us a map $\mathscr{U}(\psi):\W(\Ms)\to\W(\Ns)$ with the property that for any $H\in\H(\Ns)$, we have 
\[
	(\mathscr{U}(\psi)W)_H=\mathscr{U}_{\psi_\bullet H}(\psi)W_{\psi_\bullet H}.
\]
It is easy to show that $\W(\psi)\circ\mathscr{U}(\psi)=\text{id}_{\W(\Ms)}$, and $\mathscr{U}(\psi)\circ\W(\psi)=\text{id}_{\W(\Ns)}$. Therefore $\mathscr{U}(\psi)=\W(\psi)^{-1}$ and so $\W$ obeys the timeslice axiom.

\subsection{Relative Cauchy evolution for the Klein-Gordon Theory}
\label{sect:kgfrce}
In order to demonstrate (or rule out) dynamical locality for $\A$ or $\W$, we must first compute the relative Cauchy evolution of an arbitrary element; this has already been done for the scalar Klein-Gordon theory  in \cite{bfv} for a different construction, and we will derive a similar expression in our formalism. We begin with the theory $\A$; we fix $\h\in H(\Ms)$ and choose two subspacetimes $\Ns^\pm\subseteq\Ms$, such that:
\begin{itemize}
	\item each $\Ns^\pm$ is an object of \Loc, and their embeddings into $\Ms$ are arrows in \Loc,
	\item	each $\Ns^\pm$ contains two disjoint Cauchy surfaces $\Sadv_\pm,\Sret_\pm$ for $\Ms$ with the property that
	$\Sadv_\pm\subseteq J_\Ms^-(\Sret_\pm)$,
	\item each $\Ns^\pm$ is disjoint from the support of $\h$, and $\Ns^\pm\subseteq J^\pm_\Ms(\supp(\h))$.
\end{itemize}
We now choose two smooth partitions of unity $\cadv_\pm+\cret_\pm=1$ for $\Ms$, with the property that $\cadv_\pm(x)=1$, $\cret_\pm(x)=0$ for $x\in J_\Ms^-(\Sadv_\pm)$, and $\cadv_\pm(x)=0,$ $\cret_\pm(x)=1$ for $x\in J_\Ms^+(\Sret_\pm)$, and define
\begin{align*}
	\zeta^\pm:\F^1(\Ms)&\to\F^1(\Ms)\\
	t&\mapsto \PM\cadv_\pm\EM t.
\end{align*}
As before, we also let
\begin{align*}
	\Zeta^\pm:\F(\Ms)&\to\F(\Ms)\\
	\sum_{n=0}^N t_n&\mapsto\sum_{n=0}^N\tenpow{(\zeta^\pm)} n t_n\qquad(t_n\in\F^n(\Ms)).
\end{align*}
Additionally, we define
\begin{align*}
	\zeta^\pm[\h]:\F^1(\Ms[\h])&\to\F^1(\Ms[\h])\\
	t&\mapsto \PMh\cadv_\pm\EMh t,
\end{align*}
and define $\Zeta^\pm[\h]:\F(\Ms[\h])\to\F(\Ms[\h])$ in an analagous way to $Z^\pm$.

Now, if we denote by $\iota^\pm$, $\iota^\pm[\h]$ the embeddings of $\Ns^\pm$ into $\Ms$ and $\Ms[\h]$ respectively, it is clear that the \Alg-arrows $\A(\iota^\pm)$, $\A(\iota^\pm[\h])$ act as
\begin{align*}
	\A(\iota^\pm)[F]_{\Ns^\pm}&=[\F(\iota^\pm)F]_\Ms,\\
	\A(\iota^\pm[\h])[F]_{\Ns^\pm}&=[\F(\iota^\pm[\h])F]_{\Ms[\h]},
\end{align*}
and for any $F\in\F(\Ns^\pm)$, $f\in\Cs(\Ms)$ we have
\[
	(\F(\iota^\pm)F)[f]=F\left[f|_{\Ns^\pm}\right]=(\F(\iota^\pm[\h])F)[f].
\]
Moreover, from lemma \ref{lem:atimeslice} we can see that the inverse arrows $\A(\iota^\pm)^{-1}$, $\A(\iota^\pm[\h])^{-1}$ act as
\begin{align*}
	\A(\iota^\pm)^{-1}[F]_{\Ms}&=[\mathscr{G}(\iota^\pm)F]_{\Ns^\pm}\\
	\A(\iota^\pm[\h])^{-1}[F]_{\Ms[\h]}&=[\mathscr{G}(\iota^\pm[\h])F]_{\Ns^\pm},
\end{align*}
where for any $f\in\Cs(\Ns^\pm)$, $F\in\F(\Ms)$ and $F'\in\F(\Ms[\h])$, we see from \eqref{eqn:Gmapdef} that
\begin{align*}
	(\mathscr{G}(\iota^\pm)F)[f]&=(\Zeta^\pm F)[\iota^\pm_\ast f],\\
	(\mathscr{G}(\iota^\pm[\h])F')[f]&=(\Zeta^\pm[\h]F')[\iota^\pm[\h]_\ast f].
\end{align*}
It follows that for any $A=[F]_\Ms\in\A(\Ms)$, we have
\begin{align*}
	\rce_\Ms[\h]A&=\A(\iota^-)\A(\iota^-[\h])^{-1}\A(\iota^+[\h])
	\A(\iota^+)^{-1}A\\
	&=\left[\F(\iota^-)\mathscr{G}(\iota^-[\h])\F(\iota^+[\h])
	\mathscr{G}(\iota^+)F\right]_{\Ms}.
\end{align*}
Now, for any $f\in\Cs(\Ms)$ and $F\in\F(\Ms)$ we have
\[
	(\F(\iota^+[\h])\mathscr{G}(\iota^+)F)[f]=(\Zeta^+F)|_{\Ns^+}
	\left[f|_{\Ns^+}\right],
\]
but since the range of $\Zeta^+$ is contained in $\iota^+(\Ns^+)$, it holds that
\[
	\F(\iota^+[\h])\circ\mathscr{G}(\iota^+)=
	\iota^+[\h]_\ast\circ(\iota^+)^\ast\circ\Zeta^+,
\]
and similarly
\[
	\F(\iota^-)\circ\mathscr{G}(\iota^-[\h])=
	\iota^-_\ast\circ\iota^-[\h]^\ast\circ\Zeta^-[\h].
\]
Explicitly, the relative Cauchy evolution of $A=[F]_\Ms$ is therefore given by $\rce_\Ms[\h]A=[B[\h]F]_\Ms,$ where
\begin{align*}
	B[\h]:\F(\Ms)&\to\F(\Ms)\\
	\sum_{n=0}^Nt_n&\mapsto\sum_{n=0}^N\tenpow{\beta[\h]}n t_n\qquad(t_n\in\F^n(\Ms)),
\end{align*}
and
\begin{align}
	\beta[\h]:\F^1(\Ms)&\to\F^1(\Ms)\notag\\
	t&\mapsto\PMh\cadv_-\EMh\PM\cadv_+\EM t.\label{eqn:betadef}
\end{align}
This definition is independent of the choice of $\cadv_\pm$, provided that the regions $\Ns^\pm$ in which they are non-constant lie strictly to the future/past of $\supp(\h)$.

\subsection{Relative Cauchy evolution for Wick Polynomials}
\label{sect:wprce}

We now calculate the relative Cauchy evolution of an element $W\in\W(\Ms)$ generated by a perturbation $\h\in H(\Ms)$. While the calculation is largely similar to the process for calculating the r.c.e.\ of an element of $\A(\Ms)$, there are some subtleties introduced by the need to specify an $H\in\H(\Ms)$ to form the algebras $\T_H(\Ms)$. We will proceed as before, fixing some $\h\in H(\Ms)$ and defining $\Ns^\pm$, $\Sadv_\pm,\Sret_\pm$ and $\iota^\pm$ and $\iota^\pm[\h]$ as in the previous subsection. The relative Cauchy evolution of an element $W\in\W(\Ms)$ is given by
\[
	\rce_\Ms[\h]W=(\W(\iota^-)\circ\mathscr{U}(\iota^-[\h])\circ\W(\iota^+[\h])\circ\mathscr{U}(\iota^+))W.
\]
But when we calculate the component corresponding to $H\in\H(\Ms)$, we see that
\begin{align}
	(\rce_\Ms[\h]W)_H&=\left(\W(\iota^-)\mathscr{U}(\iota^-[\h])\W(\iota^+[\h])\mathscr{U}(\iota^+)W\right)_H\notag\\
	&=\W_H(\iota^-)\mathscr{U}_{\tilde H_\h}(\iota^-[\h])\W_{\tilde H_\h}(\iota^+[\h])\mathscr{U}_{\cH_\h}(\iota^+)W_{\cH_\h}\label{eqn:wprcedef}
\end{align}
where for any $H\in\H(\Ms)$, the distributions $\tilde H_\h\in\H(\Ms[\h])$ and $\cH_\h\in\H(\Ms)$ are defined by
\begin{align*}
	\tilde H_\h&=\iota^-[\h]_\bullet(\iota^-)^\ast H\\
	\cH_\h&=\iota^+_\bullet\iota^-[\h]^\ast\tilde H_\h.
\end{align*}
This definition is independent of the choice of $\Ns^\pm$, as a consequence of \eqref{eqn:hbisol}.

\begin{lemma}
	\label{lem:Hdiffcomp}
	Let $\Ms$ be a spacetime, and $\h\in H(\Ms)$ a metric perturbation on $\Ms$. Suppose that $H\in\H(\Ms),$ and let $\cH_\h,\Ns^\pm$ and $\cadv_\pm$ be defined as above. Then
	\beq
		\cH_\h=\tenpow{(\cEM(\eta_s^+)'\PMh(\cadv_+)'\check{E}_{\Ms[\h]}
		(\eta_s^-)'\PM(\cadv_-)')}2H,\label{eqn:Hrcedef}
	\eeq
	where $\cadv_\pm:\Ccs(\Ms)\to\Ccs(\Ms)$ are the multiplication operators induced by the functions $\cadv_\pm\in\Cs(\Ms)$, and $\eta^\pm_s:\Cs_s(\Ms)\to\Ccs(\Ms)$  are defined as multiplication by time-compact smooth functions $\eta^\pm$ that are supported in $\Ns^\pm$, such that $\eta^\pm\equiv1$ in the region in which $\cadv_\pm$ is non-constant.
\end{lemma}
\begin{proof}
Since $H$ is a bisolution, we see from the proof of lemma \ref{lem:eetapchi} that
\[
	\tenpow{(\cEM(\eta_s^\pm)'\PM(\cadv_\pm)')}2H=H.
\]
Since $\eta^\pm$ is supported in $\Ns^\pm$, it follows that $\tenpow{((\eta^\pm_s)'\PM(\cadv_\pm)')}2H$ is supported in $(\Ns^\pm)^{\times2}$, and therefore
\[
	\tilde H_\h|_{\Ns^-}=H|_{\Ns^-}=\tenpow{{\check{E}_{\Ns^-}}}2\left.
	\left(\tenpow{((\eta_s^-)'\PM(\cadv_-)')}2H\right)\right|_{\Ns^-}.
\]
Since the action of our multiplication operators does not depend on the metric of the underlying manifold, we may also consider them as maps on the corresponding function spaces on $\Ms[\h]$; since $\tilde H_\h$ is a bisolution on $\Ms[\h]$ and $\check{E}_{\Ms[\h]}|_{\Ns^-}=\check{E}_{\Ns^-}$, it follows that 
\[
	\tilde H_\h=\tenpow{(\check{E}_{\Ms[\h]}(\eta_s^-)'\PM(\cadv_-)')}2H.
\]
A similar argument yields $\cH_\h=\tenpow{(\cEM(\eta_s^+)'\PMh(\cadv_+)')}2\tilde H_\h,$ and so \eqref{eqn:Hrcedef} is satisfied.\footnote{Note that \eqref{eqn:Hrcedef} strongly resembles the action of the map $\beta[\h]$ defined in \eqref{eqn:betadef}, albeit with $\Ns^+$ and $\Ns^-$ interchanged; indeed, if we consider the subcategory of \Loc\ containing only Cauchy arrows, we can regard $\H$ as a functor from $\Loc$ to a suitable category of distribution spaces, with $\H(\psi)H=\psi_\bullet H$. This functor can be seen to be covariant; the resemblance remarked above can be explained by noting that we may define the relative Cauchy evolution of the functor $\H$ in the same way as for a locally covariant theory; this then satisfies $\rce_\Ms^{(\H)}[\h]\cH=H$.
}
\end{proof}

\begin{lemma}
	\label{lem:Hdiffsupp}
	Let $\Ms$ be a spacetime, with a metric perturbation $\h\in H(\Ms)$. Suppose also that $H\in\H(\Ms)$, and let $\cH_\h$ be defined as above. Then $\supp(H-\cH_\h)\subseteq(J_\Ms(\supp(\h))^{\times2}).$
\end{lemma}
\begin{proof}
	Let $x\in\Ms$, with $x\notin J_\Ms^+(\supp(\h))$. Since $\supp(\h)$ is compact, we can find a choice for $\Ns^-$ with $x\in\Ns^-$. It follows that $H(x,y)=\tilde H_\h(x,y)$ for all $x\notin J_\Ms^+(\supp(\h))$. Similarly, if $x\notin J_\Ms^-(\supp(\h))$ then we can find a choice for $\Ns^+$ with $x\in\Ns^+$. Therefore $\cH_\h(x,y)=\tilde H_\h(x,y)$ for all $x\notin J_\Ms^-(\supp(\h))$. Consequently, if $x\in\supp(\h)^\perp$ then $H(x,y)=\cH_\h(x,y)$. The required result follows by symmetry of $H$.
\end{proof}

 The coherency condition on elements of $\W(\Ms)$ tells us that \eqref{eqn:wprcedef} can be expressed as
\[
	(\rce_\Ms[\h]W)_H=\W_H(\iota^-)\mathscr{U}_{\tilde H_\h}(\iota^-[\h])\W_{\tilde H_\h}(\iota^+[\h])
	\mathscr{U}_{\cH_\h}(\iota^+)\tl_{H,\cH_\h}W_H.
\]
Explicitly, we can then see from \eqref{eqn:tinvdef},\eqref{eqn:winvdef} that the relative Cauchy evolution of an element $W=(W_H)_{H\in\H(\Ms)}\in\W(\Ms)$, where each $W_H$ can be represented by $T_H\in\T_H(\Ms)$, is given by
\[
	(\rce_\Ms[\h]W)_H=[B_H[\h]\lambda_{H,\cH_\h}T_H]_\Ms,
\]
where
\begin{align*}
	B_H[\h]:\T_{\cH}(\Ms)&\to\T_H(\Ms)\\
	\sum_{n=0}^Nt_n&\mapsto\sum_{n=0}^N\tenpow{\beta[\h]}nt_n
	\qquad(t_n\in\T^n(\Ms)),
\end{align*}
with
\begin{align*}
	\beta[\h]:\T^1(\Ms)&\to\T^1(\Ms)\\
	t&\mapsto \PMh\cadv_-\EMh\PM\cadv_+\EM t
\end{align*}
as before. 

Before we proceed to the dynamical locality of $\A$ and $\W$ we will need the following results. The lemma is proved in appendix \ref{appx} (cf.\ \cite[Eqn.\ 8]{dynloc2}).
\begin{lemma}
\label{lem:rcederiv}
	Let $\Ms$ be a spacetime and let $t\in\T_H^1(\Ms)$ for some $H\in\H(\Ms)$. For any $\h\in H(\Ms)$ and $f\in\Ccs(\Ms)$, we have
	\[
		\left.\dds (\beta[s\h]t)[\EM f]\right|_{s=0}=\int_\Ms dvol_\Ms\, h_{ab}T^{ab}[\EM t,\EM f],
	\]
	where
	\begin{align*}
		T^{ab}[u,\phi]=(\nabla^{(a}u)&(\nabla^{b)}\phi)-
		\frac12g^{ab}(\nabla^cu)(\nabla_c\phi)\\
		&+\frac12m^2g^{ab}u\phi+\xi(g^{ab}\Box_\g-
		\nabla^a\nabla^b-G^{ab})(u\phi)
	\end{align*}
	for $u\in\EM\T^1(\Ms)$, $\phi\in\EM\Ccs(\Ms)$.
\end{lemma}
Note that the above expression is closely linked to the classical stress-energy tensor for the Klein-Gordon theory, which we may recover via $T^{ab}[\phi]=T^{ab}[\phi,\bar\phi]$ for a smooth classical solution $\phi$.

This result leads directly to the following:
\begin{corollary}
\label{cor:rcederivn}
Let $t_n\in\T^n_H(\Ms)$ for some $H\in\H(\Ms)$ and $f\in\Ccs(\Ms)$. Then
\[
	\left.\dds\tenpow{(\beta[s\h])}nt_n[\EM f]\right|_{s=0}=n\int_\Ms dvol_\Ms\,h_{ab}T^{ab}\left[\EM\tau^n_f,\EM f\right],
\]
where
\beq
	\label{eqn:taudef}\tau^n_f(x)=\int_{\Ms^{\times(n-1)}}d^{n-1}y\,
	t_n(x,y_1,\ldots,y_{n-1})\EM f(y_1)\cdots\EM f(y_{n-1})
\eeq
for $n\geq2$, and $\tau^1_f(x)=t_1(x).$
\end{corollary}
Note that the previous two results also apply to the elements of $\F^1(\Ms)$ and $\F^n(\Ms)$ respectively, since we can consider any element of $\F^n(\Ms)$ as an element of $\T^n_H(\Ms)$ for any $H\in\H(\Ms)$.

\section{Dynamical Locality}
\label{sect:dynlockgfwp}

\subsection{Dynamical locality of the \texorpdfstring{$\xi\neq0$}{xi /= 0} Klein-Gordon theory}
\label{sect:kgfdl}

It has already been shown in \cite{dynloc2} that the Klein-Gordon theory is dynamically local in the case when $\xi=0$ and $m\neq0$, and that it is not dynamically local when $\xi=0$ and $m=0$. We wish to show that the Klein-Gordon theory $\A$ obeys the axiom of dynamical locality in the nonminimally coupled case, when $\xi\neq0$, for both $m=0$ and $m>0$. Therefore, we pick some spacetime $\Ms$ and $O\in\mathscr{O}(\Ms)$. The algebra $\A^{\text{kin}}(\Ms;O)$ is defined to be the algebra $\A(\Ms|_O)$; we recall from lemma \ref{lem:aalgimage} that for any \Loc-arrow $\psi:\Ns\hookrightarrow\Ms$, the algebra $\A(\psi)(\A(\Ns))$ comprises elements $A=[F]_\Ms$ such that $F[\EM f]=F[0]$ for every $f\in\Ccs(\Ms)$ such that $\supp(f)\cap J_\Ms(\Ns)=\emptyset$.
It follows that $F$ represents an element of $\alpha^\text{kin}_{\Ms;O}(\Ak(\Ms;O))$ if and only if $F[\EM f]=F[0]$ for all $f\in\Ccs(\Ms)$ with support lying in $O'=(\text{cl }O)^\perp$.

We can see from \eqref{eqn:dynalgdef} and \eqref{eqn:dynloccond} that if $\Ab(\Ms;K)\subseteq\alpha^\text{kin}_{\Ms;O}(\Ak(\Ms;O))$ for each spacetime $\Ms$, $O\in\mathscr{O}(\Ms)$ and $K\in\K(\Ms;O)$, then $\A$ obeys dynamical locality. Therefore, suppose that $A\in\Ab(\Ms;K)$; from the definition it follows that $\rce_\Ms[\h]A=A$ for all $\h\in H(\Ms;K^\perp)$. Now, suppose that $A$ is represented by a functional $F\in\F(\Ms)$. This means that $B[\h]F\sim_\Ms F$ for all $\h\in H(\Ms;K^\perp)$, and consequently $B[s\h]F-F\in \J(\Ms)$ for all $s\in\mathbb{R}$ sufficiently small that $s\h\in H(\Ms;K^\perp).$ Writing $F=\sum_{n=0}^Nt_n$, with each $t_n\in\F^n(\Ms)$, we can refer to lemma \ref{lem:kgfuncpolar} to see that for $n=1,\ldots,N$, we have
\beq
	\label{eqn:kgdynloc1}
	\left(\tenpow{(\beta[s\h])}nt_n\right)[\EM f]=t_n[\EM f]
\eeq
for all $f\in\Ccs(\Ms)$ and for all $\h\in H(\Ms;K^\perp)$.

Now, for each $n\geq1$ we differentiate \eqref{eqn:kgdynloc1} with respect to $s$ and set $s=0$; by corollary \ref{cor:rcederivn}, this yields
\[
	\int_\Ms dvol_\Ms\,h_{ab}T^{ab}\left[\EM\tau^n_f,\EM f\right]=0
\]
for each $\h\in H(\Ms;K^\perp)$ and $f\in\Ccs(\Ms)$, where $\tau^n_f$ is defined as in \eqref{eqn:taudef}. It follows that for all $n\geq1$, we have
\[
	T^{ab}[\EM\tau^n_f,\EM f](x)=0
\]
for all $x\in K^\perp$.

Now consider an arbitrary point $x\in K^\perp$, and a null geodesic $u:I\to K^\perp$, where $I\subset\mathbb{R}$ is an open interval containing 0 and $u(0)=x$. Since $u$ is a null geodesic, it satisfies both $u^au^bg_{ab}=0$ and $u^a\nabla_au^b=0$, where $u^a$ is the tangent vector to $u$. For each point $q$ on the geodesic we have $u_a(q)u_b(q)T^{ab}[\EM\tau^n_f,\EM f](q)=0$, and consequently for our chosen $x\in K^\perp$ we have
\[
	(\nabla_u\EM\tau^n_f(x))(\nabla_u\EM f(x))+\xi\left(-\nabla_u^2-R_{ab}(x)u^au^b\right)\left((\EM\tau^n_f(x))
	(\EM f(x))\right)=0.
\]
Note that this is equivalent to
\begin{align}
	(1-2\xi)(\nabla_u\EM\tau^n_f(x))&(\nabla_u\EM f(x))-\xi R_{ab}u^au^b(\EM\tau^n_f(x))(\EM f(x))
	\notag\\
	&+\xi(\EM\tau^n_f(x)\nabla_u^2\EM f(x)+\EM f(x)\nabla_u^2\EM\tau^n_f(x))=0.
	\label{eqn:kgst2}
\end{align}
It follows that for any $f\in\Ccs(\Ms)$ for which $\EM f(x)=0=\nabla_u\EM f(x)$ and $\nabla_u^2\EM f(x)\neq 0$,\footnote{Such a solution always exists; we may explicitly construct one as follows. We work in normal coordinates $q^a$ in a neighbourhood $S\ni x$ such that $x$ is at the origin, and the $q^0=0$ hyperplane is a subset of a spacelike Cauchy surface $\Sigma\subset\Ms$, and we take our null geodesic $u$ such that in coordinates, the tangent at $x$ is $u^a(x)=(1,1,0,\ldots,0).$ Then any solution $\psi$ is uniquely determined by its data $(\varphi,\pi)$ on $\Sigma$, where $\varphi(\underline{q})=\psi|_\Sigma(\underline{q})$ and $\pi(\underline{q})=(\nabla_0\psi)|_\Sigma(\underline{q})$. It is then easy to check that defining $\varphi(\underline{q})=(q^1)^2$, $\pi(\underline{q})=0$ for $\underline{q}\in\Sigma\cap S$ gives us a solution $\psi$ satisfying the above conditions.} we have $\EM\tau^n_f(x)=0$, as $\xi\neq0$.

In the case that $n=1$, we have $\EM\tau^1_f=\EM t_1$ for all $f$, so we immediately see that $\EM t_1(x)=0$ for all $x\in K^\perp$. Now, we look at the case where $n=2$. We have $\EM\tau^2_f(x)=\int_\Ms dy\,t_2(x,y)\EM f(y)$, which is linear in $f$. Let $f$ be chosen such that $\EM f(x)=0=\nabla_u\EM f(x)$ and $\nabla_u^2\EM f(x)\neq 0$; additionally, we choose $f'\in\Ccs(\Ms)$ such that $\supp(f')\subset \{x\}^\perp$. Then $\EM f+\EM f'=\EM f$ in an open neighbourhood of $x$, so 
\[
	\EM\tau^2_{f'}(x)=\EM\tau^2_{f+f'}(x)-\EM\tau^2_f(x)=0.
\]
It follows that for any $f'\in\Ccs(\Ms)$ supported outside $J_\Ms(x)$, we have
\[
	\int_\Ms dy\,(\tenpow\EM2t_2)(x,y)f'(y)=-\EM\tau^2_{f'}(x)=0.
\]
Therefore $\tenpow{\EM}2t_2(x,y)=0$ whenever $x\in K^\perp$ and $y\in \{x\}^\perp$.

However, by the definition of $\F^2(\Ms)$, we have $\tenpow\EM2t_2(x,\cdot)\in\EM\Ccs(\Ms)$ for any fixed $x\in\Ms$, and it is therefore a smooth classical Klein-Gordon solution. If $\Sigma$ is a spacelike Cauchy surface containing $x$, then the data for $\tenpow\EM2t_2(x,\cdot)$ on $\Sigma$ is supported in $\{x\}$ for any $x\in K^\perp$ by the above result. But the data for a smooth solution is itself smooth, and therefore cannot be both nonzero and supported at a point. Consequently $\tenpow\EM2t_2(x,y)=0$ for any $(x,y)\in K^\perp\times\Ms$, and by symmetry we have $\supp(\tenpow\EM2 t_2)\subseteq J_\Ms(K)^{\times 2}$.

Now, consider the case where $n>2$. Suppose that we have $f,f_1$ such that $\EM f(x)=0=\nabla_u\EM f(x)$, $\EM f_1(x)=0=\nabla_u\EM f_1(x)$, and $\EM f(x)\neq0$. Then, for sufficiently small $\kappa$ we have $\EM\tau^n_{f+\kappa f_1}(x)=0$. Therefore, by symmetry of $t_n$ we have
\begin{align*}
	\EM\tau^n_{f}(x)+(-1)^{n-1}(n-1)\kappa\int_{\Ms^{\times(n-1)}}d^{n-1}y\,\Big[(\tenpow\EM n &t_n)(x,y_1,\ldots,y_{n-1})\\
	&f_1(y_1) f(y_2)\cdots f(y_{n-1})\Big]
	+\mathcal{O}(\kappa^2)=0.
\end{align*}
Differentiating this expression with respect to $\kappa$ and setting $\kappa=0$, we have
\[
	\int_{\Ms^{\times(n-1)}}d^{n-1}y\,(\tenpow\EM n t_n)(x,y_1,\ldots,y_{n-1}) f_1(y_1)f(y_2)\cdots f(y_{n-1})=0.
\]
We may repeat this argument to see that
\[
	\int_{\Ms^{\times(n-1)}}d^{n-1}y\,(\tenpow\EM n t_n)(x,y_1,\ldots,y_{n-1})f_1(y_1)\cdots f_{n-1}(y_{n-1})=0
\]
for any $f_1,\ldots,f_{n-1}$ such that $\EM f_i(x)=0=\nabla_u\EM f_i(x)$, $i=1,\ldots,n-1$. It follows that for any $x_1\in K^\perp$, we have $\tenpow\EM n t_n(x_1,\ldots,x_n)=0$ whenever at least one of $x_2,\ldots,x_n$ lies in ${x_1}^\perp$. Fixing $x_1\in K^\perp$, we note that $\tenpow\EM nt_n(x_1,y_1,\ldots,y_{n-1})$ is a smooth Klein-Gordon $(n-1)$-solution; its data on a spacelike Cauchy surface $\Sigma\ni x$ is supported in $\{x\}^{\times(n-1)}.$ Consequently we must have $\tenpow\EM nt_n(x_1,y_1,\ldots,y_{n-1})=0$ for $x_1\in K^\perp,\ y_1,\ldots,y_{n-1}\in\Ms$ by smoothness. Therefore we have proved the following lemma:
\begin{lemma}
	\label{lem:kgsuppprop}
	Let $\Ms$ be a spacetime and let $t_n\in\F^n(\Ms)$, $n\geq1.$ If $O\in\mathscr{O}(\Ms)$, $K\in \mathscr{K}(\Ms;O)$ and $\left(\tenpow{(\beta[s\h])}nt_n\right)[\EM f]=t_n[\EM f]$ for all $f\in\Ccs(\Ms)$ and for all $\h\in H(\Ms;K^\perp)$, then
	\[
		\supp(\tenpow\EM nt_n)\subseteq J_\Ms(K)^{\times n}.
	\]
\end{lemma}
From here we may prove the following result:
\begin{theorem}
	The Klein-Gordon theory is dynamically local in the nonminimally coupled case, for all $m\geq0$.
\end{theorem}
\begin{proof}
	Recall that for any spacetime $\Ms$ and $O\in\mathscr{O}(\Ms)$, the algebra $\alpha^{\text{kin}}_{\Ms;O}(\Ak(\Ms;O))$ comprises elements represented by functionals $F$ with the property that $F[\EM f]=F[0]$ for all $f\in\Ccs(\Ms)$ supported within $O'$. To demonstrate that the theory is dynamically local, it is sufficient to show that $\Ab(\Ms;K)\subseteq\alpha^{\text{kin}}_{\Ms;O}(\Ak(\Ms;O))$ for all $K\in\mathscr{K}(\Ms;O)$. Given such a $K$, and an element $A\in\Ab(\Ms;K)$ represented by $F=\sum_{n=0}^Nt_n$, with each $t_n\in\F^n(\Ms)$, we may see from \eqref{eqn:kgdynloc1} and lemma \ref{lem:kgsuppprop} that $\supp(\tenpow\EM nt_n)\subseteq J_\Ms(K)^{\times n}$ for each $n=1,\ldots,N$, and subsequently $t_n[\EM f]=0$ for all $f\in\Ccs(K^\perp)$. Therefore in particular we have $F[\EM f]=t_0=F[0]$ for all $f\in\Ccs(O')$, and so $F$ represents an element of $\alpha^{\text{kin}}_{\Ms;O}(\Ak(\Ms;O))$. Consequently the theory is dynamically local.
\end{proof}

\subsection{Dynamical Locality of the algebra of Wick Polynomials}
We now proceed to examine the cases in which we can demonstrate dynamical locality for the theory $\W$. We begin by looking at the minimally coupled massless case. The corresponding case for the Klein-Gordon theory is not dynamically local, and so one would not expect dynamical locality to hold here. Indeed, this is the case; when $\xi=m=0$, any constant function is a classical solution to the Klein-Gordon equation. Therefore, in any spacetime $\Ms$ with compact Cauchy surfaces, the function $\phi(x)=1$ is an element of $\EM\Ccs(\Ms)$. However, we have $ T^{ab}[\phi,\EM f]=0$; it follows that for any $t\in\T^1(\Ms)$ such that $\EM t\equiv1$, we have $t\in\Wd(\Ms;O)$ for any $O\in\mathscr{O}(\Ms)$. But it is also the case that if we pick $f\in\Ccs(O')$ with $\int_\Ms dx\,f(x)\neq0$, then $t[\EM f]\neq0$; therefore, $t\notin\Wk(\Ms;O)$.

We may, however, demonstrate dynamical locality in two cases. To do this, we need the following results:
\begin{lemma}
	\label{lem:tsuppprop}
	Let $\Ms$ be a spacetime with $O\in\mathscr{O}(\Ms)$ and $K\in\mathscr{K}(\Ms;O)$. Let $t_n\in\T^n(\Ms)$ for some $n\geq0$, and suppose that for all $f\in\Ccs(\Ms)$ and $\h\in H(\Ms;K^\perp)$ we have
	\beq
		\label{eqn:tintzero}
		\int_\Ms dvol_\Ms\,h_{ab}T^{ab}[\EM\tau^n_f,\EM f]=0,
	\eeq
	where $\tau^n_f$ is defined as in \eqref{eqn:taudef}. Then, in the massive minimally coupled and massive conformally coupled cases, we have 
	$\supp(\tenpow\EM nt_n)\subseteq J_\Ms(K)^{\times n}$.
\end{lemma}
\begin{proof}
	We will consider the massive minimally coupled case first, in which $m\neq0$ and $\xi=0$. Clearly $T^{ab}[\EM\tau^n_f,\EM f](x)=0$ for all $f\in\Ccs(\Ms)$ and $x\in K^\perp$; now, we fix $x\in K^\perp$ and pick some $f\in\Ccs(\Ms)$ such that $(\EM f)(x)\neq 0$. In the case where $\Ms$ has dimension 2, we note that 
\[
	0=g_{ab}T^{ab}[\EM\tau^n_f,\EM f](x)=m^2\EM\tau^n_f(x)\EM f(x),
\]
and consequently $\EM\tau^n_f(x)=0$ for any such $f$; in higher dimensions, we choose normal coordinates at $x$ oriented such that $\nabla_2\EM f(x)=\cdots=\nabla_{d-1}\EM f(x)=0$, and define $v_{ab}$ such that in these coordinates we have $v_{00}=1,\ v_{11}=-1$, and all other entries zero. It follows that $v_{ab}g^{ab}(x)=2$ and $v_{ab}\nabla^{(a}\EM\tau^n_f(x)\nabla^{b)}\EM f(x)=\nabla^a\EM\tau^n_f(x)\nabla_a\EM f(x)$, so that we have
\[
	0=v_{ab}T^{ab}[\EM\tau^n_f,\EM f](x)=m^2\EM\tau^n_f(x)\EM f(x).
\]
Again, we may conclude that $\EM\tau^n_f(x)=0$ for any such $f$.

When $n=1$ we deduce immediately that $\EM t_1(x)=0$ for all $x\in K^\perp$. For $n=2$, we note that $\tau^2_f$ is linear in $f$, and as any $f\in\Ccs(\Ms)$ may be expressed as $f=f_1-f_2$ where $\EM f_1(x)\neq0\neq\EM f_2(x)$ we have $\EM\tau^2_f(x)=-\int_\Ms dy\,(\tenpow\EM2 t_2)(x,y)f(y)=0$ for all $f\in\Ccs(\Ms)$. Therefore $\tenpow\EM2 t_2(x,y)=0$ for all $x\in K^\perp$, and so $\supp(\tenpow\EM2 t_2)\subseteq J_\Ms(K)^{\times 2}$ by symmetry. For $n>2$, we pick $f\in\Ccs(\Ms)$ with $\EM f(x)\neq 0$ and let $f_1\in\Ccs(\Ms)$ be arbitrary; for sufficiently small $\kappa$ we have $\EM \tau^n_{f+\kappa f_1}(x)=0$. We differentiate this expression with respect to $\kappa$ and set $\kappa=0$, which yields
\[
	\int_{\Ms^{\times(n-1)}}d^{n-1}y\,(\tenpow\EM n t_n)(x,y_1,\ldots,y_{n-1})f_1(y_1)f(y_2)\cdots f(y_{n-1})=0;
\]
we may then repeat this argument to see that
\[
	\int_{\Ms^{\times(n-1)}}d^{n-1}y\,(\tenpow\EM n t_n)(x,y_1,\ldots,y_{n-1})f_1(y_1)f_2(y_2)\cdots f_{n-1}(y_{n-1})=0
\]
for any $f_1,\ldots,f_{n-1}\in\Ccs(\Ms)$. It follows that $\tenpow\EM n t_n(x,y_1,\ldots,y_{n-1})=0$ for all $x\in K^\perp$, and by symmetry we have $\supp(\tenpow\EM n t_n)\subseteq J_\Ms(K)^{\times n}$. This concludes the proof for the massive minimally coupled case.

In the massive conformally coupled case, where $m\neq0$ and $\xi=\frac{d-2}{4(d-1)}$, where $d$ is the dimension of $\Ms$, we have $g_{ab}T^{ab}[\phi_1,\phi_2]=m^2\phi_1\phi_2$ for any $\phi_1,\phi_2\in\EM\Ccs(\Ms)$. It follows that for all $x\in K^\perp$, we have $\EM\tau^n_f(x)\EM f(x)$ for all $f\in\Ccs(\Ms)$. We may use the same argument as above to show that $\supp(\tenpow\EM n t_n)\subseteq J_\Ms(K)^{\times n}.$
\end{proof}

\begin{lemma}
	\label{lem:Esupport}
	Let $t_n\in\T^n(\Ms)$, and suppose that $\supp(\tenpow\EM nt_n)\subseteq J_\Ms(K)^{\times n}$. Furthermore, let $S$ be any open neighbourhood of an arbitrary Cauchy surface $\Sigma\subset\Ms$. Then there exist $s,u_k\in\T^n(\Ms)$, $k=1,\ldots,n$, such that
	\[
		t_n=s+\sum_{k=1}^n(\PM)_ku_k,
	\]
	where we define $(\PM)_k=\tenpow\id{k-1}\otimes\PM\otimes\tenpow\id{n-k}$, and such that  $\supp(s)\subseteq (J_\Ms(K)\cap S)^{\times n}$.
\end{lemma}
\begin{proof}
To prove this, we will need the result of lemma \ref{lem:Ekernel}: namely, that
\[
	\ker\tenpow\EM n=\left\{\sum_{k=1}^n(\PM)_ku_k:u_k\in\T^n(\Ms)\right\}.
\]
Now, if $S$ is an open neighbourhood of a Cauchy surface, then we can find two disjoint Cauchy surfaces $\Sigma^\pm\subset S$ such that $\Sigma^+\subset J^+_\Ms(\Sigma^-)$. Let $\cadv+\cret=1$ be a smooth partition of unity such that $\cadv(x)=0,$ $\cret(x)=1$ for $x\in J_\Ms^+(\Sigma^+)$ and $\cadv(x)=1$, $\cret(x)=0$ for $x\in J_\Ms^-(\Sigma_-)$. We let $s=\tenpow{(\PM\cadv\EM)}nt_n$; by \eqref{eqn:epchidist} we have $\tenpow\EM ns=\tenpow\EM nt_n$, so by lemma \ref{lem:Ekernel} it follows that
\[
	t_n-s=\sum_{k=1}^n(\PM)_ku_k
\]
for some $u_k\in\T^n(\Ms),$ $k=1,\ldots,n$. The required support properties of $s$ follow from the support of $\tenpow\EM nt_n$ and the fact that $\cadv$ is constant outside $S$.
\end{proof}

The above results allow us to prove the following:
\begin{theorem}
	The theory $\W$ of Wick polynomials is dynamically local in the massive minimally coupled case and the massive conformally coupled case. The theory is not dynamically local in the massless minimally coupled case.
\end{theorem}
\begin{proof}We pick a spacetime $\Ms$, and some $O\in\mathscr{O}(\Ms)$; we will denote the dynamical and kinematic nets for $\W$ by $\wdyn_{\Ms;O}$, and $\wkin_{\Ms;O}$ respectively. We may then use a similar argument to that used above to see that $\wkin_{\Ms;O}(\W^\text{kin}(\Ms;O))$ comprises elements $W\in\W(\Ms)$ with $W=[W_H]_{H\in\H(\Ms)}$, where each $W_H\in\W_H(\Ms)$ can be represented by $T_H\in\T_H(\Ms)$ with the property that $T_H[\EM f]=T_H[0]$ for all $f\in\Ccs(O')$.

As already mentioned, for an additive theory, it is sufficient for dynamical locality to show that we have $\Wb(\Ms;K)\subseteq\wkin_{\Ms;O}(\W^\text{kin}(\Ms;O))$ for all $K\in\mathscr{K}(\Ms;O)$; we therefore pick some such $K$, and let $W\in\Wb(\Ms;K)$. Let $W=(W_H)_{H\in\H(\Ms)}$, and pick some fixed $H\in\H(\Ms)$; moreover, let $W_H\in\W_H(\Ms)$ be represented by a functional $T_H\in\T_H(\Ms)$. Since $\rce_\Ms[\h]W=W$ for all $\h\in H(\Ms;K^\perp)$ it follows that
\beq
	\label{eqn:wpfunccond}
	B_H[\h]\lambda_{H,\cH_\h}T_H\sim_\Ms T_H
\eeq
for all such $\h$. If $T_H=\sum_{n=0}^Nt_n$ with each $t_n\in\T^n(\Ms)$, then using \eqref{eqn:lambdadef}, interchanging sums and relabelling, we may write
\begin{align*}
	\lambda_{H,\cH_\h}T_H&=\sum_{k=0}^{\lfloor n/2\rfloor}\frac{1}{k!}\sum_{n=0}^N\left\langle\tenpow{(H-\cH_\h)}k,
	t_n^{(2k)}\right\rangle\\
	&=\sum_{n=0}^N\sum_{k=0}^{\lfloor\frac{N-n}2\rfloor}\frac1{k!}\left\langle\tenpow{(H-\cH_\h)}k,t_{n+2k}^{(2k)}\right\rangle;
\end{align*}
the precise meaning of the notation here is given in \eqref{eqn:anglefuncdef}.
Note that in the second sum, the inner sum for each $n$ consists only of elements of $\T^n(\Ms)$; we write
\beq
	\label{eqn:wptildedef}
	\T^n(\Ms)\ni\tilde t_{n;\h}=\sum_{k=0}^{\lfloor\frac{N-n}2\rfloor}\frac1{k!}\left\langle\tenpow{(H-\cH_\h)}k,t_{n+2k}^{(2k)}\right\rangle
\eeq
for $n=0,\ldots,N$, and may express the condition \eqref{eqn:wpfunccond} as 
\beq
	\label{eqn:wpfunccond2}
	\left(\tenpow{(\beta[\h])}n\tilde t_{n;\h}\right)[\EM f]=t_n[\EM f]
\eeq
for all $f\in\Ccs(\Ms)$ and for each $1\leq n\leq N$. We note that the $n=0$ term in \eqref{eqn:wpfunccond} requires
\beq
	\label{eqn:wpfunccond0}
	\sum_{k=1}^{\lfloor N/2\rfloor}\frac1{k!}\left\langle\tenpow{(H-\cH_\h)}k,t_{2k}^{(2k)}\right\rangle=0
\eeq
for all $\h\in H(\Ms;K^\perp).$

It follows from \eqref{eqn:wpfunccond2} that for $n\geq1$  we have
\[
	\left.\frac{d}{ds}\left(\tenpow{(\beta[s\h])}n\tilde t_{n;s\h}\right)[\EM f]\right|_{s=0}=0
\]
for all $f\in\Ccs(\Ms)$ and $\h\in H(\Ms;K^\perp)$. But since $\beta[\mathbf{0}]=\id$ and $\tilde t_{n;\mathbf{0}}=t_n$, this is equivalent to
\beq
	\label{eqn:sepwpdiffcond}
	\left.\frac{d}{ds}\left(\tenpow{(\beta[s\h])}nt_n\right)[\EM f]\right|_{s=0}+\left.\frac{d}{ds}\tilde t_{n;s\h}[\EM f]\right|_{s=0}=0;
\eeq
by corollary \ref{cor:rcederivn}, we have
\beq
	\left.\frac{d}{ds}\left(\tenpow{(\beta[s\h])}nt_n\right)[\EM f]\right|_{s=0}=n\int_\Ms dvol_\Ms\,h_{ab}T^{ab}[\EM \tau^n_f,\EM f],
	\label{eqn:betaderivdef}
\eeq
where as before $\tau^n_f$ is defined according to \eqref{eqn:taudef}.

We now wish to show that in fact $\tilde t_{n;\h}\sim_\Ms t_n$ for all $\h\in H(\Ms;K^\perp)$ and $n\geq0$. To do so, we firstly note that by \eqref{eqn:wptildedef}, we automatically have $\tilde t_{N;\h}=t_N$ and $\tilde t_{N-1;\h}=t_{N-1}$ for all $\h\in H(\Ms)$. We may then proceed by descent, using the fact that $\tilde t_{n;\h}\sim_\Ms t_n$ for all $\h\in H(\Ms;K^\perp)$ if $\tilde t_{n+2k;\h}\sim_\Ms t_{n+2k}$ for all $k$ satisfying $2\leq2k\leq N-n$. This can be shown from the previous results, as follows.

If $\tilde t_{n+2k;\h}\sim_\Ms t_{n+2k}$ for $2\leq2k\leq N-n$, then (with $n$ replaced by $n+2k$) the second term on the left hand side of \eqref{eqn:sepwpdiffcond} vanishes, and so by \eqref{eqn:betaderivdef} we also have
\[
	\int_{\Ms}dvol_\Ms\,h_{ab}T^{ab}[\EM \tau^{n+2k}_f,\EM f]=0
\]
for $2\leq2k\leq N-n$ and $\h\in H(\Ms;K^\perp)$. It follows from lemma \ref{lem:tsuppprop} that in the massive minimally coupled and massive conformally coupled theories, we have $\supp(\tenpow\EM{(n+2k)} t_{n+2k})\subseteq J_\Ms(K)^{\times(n+2k)}$. We may now use lemma \ref{lem:Esupport} to see that for any open neighbourhood $S$ of an arbitrary Cauchy surface, the distributions $t_{n+2k}$ may be written
\[
	t_{n+2k}=s+\sum_{j=1}^{n+2k}(\PM)_ju_j
\]
where $s,u_j\in\T^{n+2k}(\Ms)$ and $\supp(s)\subseteq (J_\Ms(K)\cap S)^{\times(n+2k)}$. If we now fix some $\h\in H(\Ms,K^\perp)$ and choose $S$ such that $J_\Ms(\supp(\h))\cap J_\Ms(K)\cap S=\emptyset$, it follows that
\[
	\left\langle\tenpow{(H-\cH_\h)}k,s^{(2k)}\right\rangle=0, 
\]
recalling from lemma \ref{lem:Hdiffsupp} that $\supp(H-\cH_\h)\subseteq(J_\Ms(\supp(\h)))^{\times2}$. But this means that for all $f\in\Ccs(\Ms)$, we have
\[
	\left\langle\tenpow{(H-\cH_\h)}k,t_{n+2k}^{(2k)}\right\rangle[\EM f]=
 	\sum_{j=1}^{n+2k}\left\langle\tenpow{(H-\cH_\h)}k,(\PM)_ju_j^{(2k)}
 	\right\rangle[\EM f]=0
\]
for $2\leq2k\leq N-n$, where we have used the fact that $(\PM\otimes\id)H=0=(\id\otimes\PM)H$ for any $H\in\H(\Ms)$. By \eqref{eqn:wptildedef}, we therefore have $\tilde t_{n;\h}\sim_\Ms t_n$.

As observed above, we certainly know that $\tilde t_{N;\h}\sim_\Ms t_N$ and $\tilde t_{N-1;\h}\sim_\Ms t_{N-1}$ for all $\h\in H(\Ms;K^\perp)$, and consequently by the above arguments $\tilde t_{N-2;\h}\sim_\Ms t_{N-2}$ and $\tilde t_{N-3;\h}\sim t_{N-3}$ for all $\h\in H(\Ms;K^\perp)$. We may continue this argument to see that in fact $\tilde t_{n;\h}\sim_\Ms t_n$ for all $n\geq0$ and $\h\in H(\Ms;K^\perp)$. Therefore \eqref{eqn:sepwpdiffcond} and \eqref{eqn:betaderivdef} tell us that \eqref{eqn:tintzero} is satisfied for all $n\geq1$; a final use of lemma \ref{lem:tsuppprop} tells us that $\supp(\tenpow\EM nt_n)\subseteq J_\Ms(K)^{\times n}$ for $n=1,\ldots,N$.

This firstly shows that the condition \eqref{eqn:wpfunccond0} is satisfied. More importantly, it shows that $t_n[\EM f]=0$ for all $f\in\Ccs(K^\perp)$ and $n\geq 1$, and therefore if $T$ represents an element of $\Wb(\Ms;K)$, then $T[\EM f]=t_0=T[0]$ for all $f\in\Ccs(K^\perp)$. If this is the case for all $K\in\mathscr{K}(\Ms;O)$ then $T$ represents an element of $\wkin_{\Ms;O}(\W^\text{kin}(\Ms;O))$, and so $\Wb(M;K)\subseteq\wkin_{\Ms;O}(\W^\text{kin}(\Ms;O))$ for all $K\subset O$. Therefore the massive minimally coupled and massive conformally coupled theories are dynamically local. We have already observed that the massless minimally coupled theory is not dynamically local.

\end{proof}

\subsubsection*{Acknowledgement} The author wishes to thank Chris Fewster for help, support and many useful conversations throughout the course of this work.

\appendix
\section{Appendix}
\label{appx}

%

\begin{lemma}
\label{lem:rcederivappx}
	Let $\Ms$ be a locally covariant theory obeying the timeslice axiom and let $t\in\T_H^1(\Ms)$ for some $H\in\H(\Ms)$. For any $\h\in H(\Ms)$ and $f\in\Ccs(\Ms)$, we have
	\[
		\left.\dds (\beta[s\h]t)[\EM f]\right|_{s=0}=\int_\Ms dvol_\Ms\, h_{ab}T^{ab}[\EM t,\EM f],
	\]
	where
	\begin{align*}
		T^{ab}[u,\phi]=(\nabla^{(a}u)&(\nabla^{b)}\phi)-\frac12g^{ab}(\nabla^cu)(\nabla_c\phi)\\
		&+\frac12m^2g^{ab}u\phi+\xi(g^{ab}\Box_\g-\nabla^a\nabla^b-G^{ab})(u\phi)
	\end{align*}
	for $u\in\EM\T^1(\Ms)$, $\phi\in\EM\Ccs(\Ms)$.
\end{lemma}
\begin{proof} We adopt and adapt the strategy used in Appendix~B of~\cite{dynloc2}.
	Let $\h\in H(\Ms)$, and consider the metric perturbation $s\h$ where $s\in\mathbb{R}$ is sufficiently small to ensure that $s\h\in H(\Ms)$. We have $\EM\beta[s\h]t=\EM\zeta^-[s\h]\zeta^+t$, therefore
\[
	\EM\beta[s\h]t-\EM t=\EM(\PMsh-\PM)\cadv_-\EMsh\zeta^+t+\EM\PM\cadv_-(\EMsh-\EM)\zeta^+t.
\]
Since $\PM$ is a differential operator it follows that the support of $(\PMsh-\PM)f$ lies within $\supp(\h)\cap\supp(f)$ for any $f\in\Cs(\Ms)$. The support of $\cadv_-$ lies strictly to the past of $\supp(\h)$, so the first term above vanishes. Moreover, note that the support of $\EMsh^+f-\EM^+f$ is contained within $J_\Ms^+(\supp(\h))$ for any $f\in\Ccs(\Ms)$, and is also therefore disjoint from $\supp(\cadv_-)$; it follows that
\[
	\EM\beta[s\h]t-\EM t=\EM\PM\cadv_-(\EMsh^--\EM^-)\zeta^+t.
\]
	Similarly, $(\EMsh^- -\EM)f$ must be supported in $J_\Ms^-(\supp(\h))$, for any $f\in\Ccs(\Ms)$; it follows that the support of $\cret_-(\EMsh^--\EM)f$ is compact. Therefore
\[
	\EM\beta[s\h]t-\EM t=\EM\PM(\EMsh^--\EM^-)\zeta^+t.
\]
We use $\PMsh\EMsh^-\zeta^+t=\zeta^+t=\PM\EM^-\zeta^+t$ to see that
\begin{align*}
	\EM\beta[s\h]t-\EM t&=-\EM(\PMsh-\PM)\EMsh^-\zeta^+t\\
	&=\EM(\PMsh-\PM)\EMsh^-(\PMsh-\PM)\EM^-\zeta^+t\\
	&\hspace{140pt}-\EM(\PMsh-\PM)\EM^-\zeta^+t,
\end{align*}
where we have used the fact that $\EMh^-\PMh\EM^-u=\EM^-u$ for any $u\in\mathcal{E}'(\Ms)$ and $\h\in H(\Ms)$; this is proved below.

Finally, we note that $\supp(\h)\cap\supp(\EM^+\zeta^+t)=\emptyset$, so
\begin{align}
	\EM\beta[s\h]t-\EM t&=\EM(\PMsh-\PM)\EMsh^-(\PMsh-\PM)\EM^-\zeta^+t\notag\\
	&\hspace{140pt}-\EM(\PMsh-\PM)\EM\zeta^+t.\label{eqn:Ebetadiff}
\end{align}

Now, for any $f\in\Cs(\Ms)$ we have
\[
	\Box_{\g+s\h} f=\Box_\g f+s\left(\frac12\nabla^a({h^b}_b)\nabla_af-\nabla^a(h_{ab}\nabla^bf)\right)
	+\mathcal{O}(s^2);
\]
we may also note that
\[
	\left.\dds R_{\g+s\h}\right|_{s=0}=(g^{ab}\Box_\g-\nabla^a\nabla^b-R^{ab})h_{ab}
\]
(see e.g.\ \cite{osterbrink}). It follows that $\lim_{s\to 0}((\PMsh-\PM)/s)f$ exists and is equal to
\begin{align*}
	&\bigg(\frac12\nabla^a({h^b}_b\nabla_a f)+\frac12{h^b}_bm^2f
	+\frac12{h^b}_b\xi R_\g f\\
	&\hspace{20mm}-\nabla^a(h_{ab}\nabla^b f)+\xi f\left(g^{ab}\Box_\g-\nabla^a\nabla^b-R^{ab}\right)h_{ab}\bigg).
\end{align*}
By duality, the same limit holds for distributions in the weak topology. Moreover, the first term of \eqref{eqn:Ebetadiff} can now be seen to be of order $\mathcal{O}(s^2)$ as $s\to0$, and therefore
\begin{align*}
	\left.\dds\EM\beta[s\h]t\right|_{s=0}&=-\EM\bigg(\frac12\nabla^a({h^b}_b\nabla_a\EM t)+\frac12{h^b}_bm^2\EM t+\frac12{h^b}_b\xi R_\g\EM t\\
	&\hspace{20mm}-\nabla^a(h_{ab}\nabla^b\EM t)+\xi\EM t\left(g^{ab}\Box_\g-\nabla^a\nabla^b-R^{ab}\right)h_{ab}\bigg),
\end{align*}
where the derivative is taken in the weak topology.
We can now see that for any $f\in\Ccs(\Ms)$, we have
\begin{align*}
	\left.\dds(\beta[s\h]t)[\EM f]\right|_{s=0}&=\int_\Ms dvol_\Ms\,(\EM f)\bigg(\frac12\nabla^a({h^b}_b\nabla_a\EM 
	t)+\frac12{h^b}_bm^2\EM t\\
	&\hspace{30mm}+\frac12{h^b}_b\xi R_\g\EM t-\nabla^a(h_{ab}\nabla^b\EM t)\\
	&\hspace{40mm}+\xi\EM t\left(g^{ab}\Box_\g-\nabla^a\nabla^b-R^{ab}\right)h_{ab}\bigg).
\end{align*}

Integration by parts then yields
\[
	\left.\dds(\beta[s\h]t)[\EM f]\right|_{s=0}=\int_\Ms dvol_\Ms\,h_{ab}T^{ab}[\EM t,\EM f]
\]
as required.

It remains to show that for any $u\in\mathcal{E}'(\Ms)$ and $\h\in H(\Ms)$ we have $\EMh^-\PMh\EM^-u=\EM^-u$. We may see that this holds by considering an arbitrary $f\in\Ccs(\Ms)$ and splitting $\EM^-u=t+t'$ where $t\in\mathcal{E}'(\Ms)$ and $t'\in\mathcal{D}'(\Ms)$ with $J_{\Ms[\h]}^-(\supp(t'))\cap\supp(f)=\emptyset$. It follows that
\begin{align*}
	\EMh^-\PMh\EM^-u[f]&=\EMh^-\PMh t[f]+\EMh^-\PMh t'[f]\\
	&=t[f].
\end{align*}
But $t[f]=t[f]+t'[f]=\EM^-u[f]$. Since $f$ was arbitrary, we have $\EMh^-\PMh\EM^-u=\EM^-u$.
\end{proof}

\begin{lemma}
	\label{lem:Ekernel}
	Let $\Ms$ be a spacetime, and consider $\EM$ as a map from $\T^1(\Ms)$ to $\mathcal{D}'(\Ms)$. Then for all $n\in\mathbb{N}$,
	\[
		\ker\tenpow\EM n=\left\{\sum_{k=1}^n(\PM)_ku_k:u_k\in\T^n(\Ms)\right\},
	\]
	where $(\PM)_k=\tenpow\id{k-1}\otimes\PM\otimes\tenpow\id{n-k}$.
\end{lemma}
\begin{proof}
Let $S_n$ denote the set in the right hand side of the above equation. Clearly any distribution in $S_n$ lies in $\ker\tenpow\EM n$; therefore, we need only prove the inclusion $\ker\tenpow\EM n\subseteq S_n$. Suppose that $t\in\T^1(\Ms)$, with $\EM t=0$. We have $\EM^+t=\EM^-t$; $t$ is compactly supported, and so by the support properties of $\EM^\pm t$ we must have that $\EM^+t$ is compactly supported. But $t=\PM\EM^+t$, and therefore $\ker\EM\subseteq\PM\T^1(\Ms)$. This proves the case where $n=1$. Now suppose that $t_n\in\T^n(\Ms)$ with $\tenpow\EM nt_n=0$. Then pick two disjoint Cauchy surfaces $\Sigma^\pm\subset\Ms$ with $\Sigma^+\subset J_\Ms^+(\Sigma^-)$, and some $\chi\in\Cs(\Ms)$ with $\chi(x)=0$ for $x\in J_\Ms^+(\Sigma^+)$ and $\chi(x)=1$ for $x\in J_\Ms^-(\Sigma^-)$. We know that $\EM\PM\chi\EM t=\EM t$ and that $\PM\chi\EM t$ is compactly supported for any $t\in\T^1(\Ms)$; it follows that
\[
	(\EM^+\otimes\tenpow{(\PM\chi\EM)}{n-1})t_n=
	(\EM^-\otimes\tenpow{(\PM\chi\EM)}{n-1})t_n,
\]
and that by the support properties given above the left hand side of the above equation must be compactly supported. Denoting this as $u_1$ we therefore have $t_n=(\PM)_1u_1+v_1$, where
\[
	v_1=t_n-\id\otimes\tenpow{(\PM\chi\EM)}{n-1}t_n.
\]
As observed earlier, since $u_1$ is obtained from an element of $\T^n(\Ms)$ by the application of $\EM^\pm$, differential operators and multiplication by smooth functions, it follows that its wavefront set also has the desired properties for $u_1$ itself to be an element of $\T^n(\Ms)$. 

Now $t_n,(\PM)_1u_1\in\T^n(\Ms),$ so $v_1\in\T^n(\Ms)$; but $\id\otimes\tenpow\EM{n-1}v_1=0$, so $v_1\in\ker(1\otimes\tenpow\EM{n-1})$: we may repeat the argument to see that $v_1=(\PM)_2u_2+v_2$ for some $u_2,v_2\in\T^n(\Ms)$ with $v_2\in\ker(1\otimes1\otimes\tenpow\EM{n-2})$. Continuing the argument further, we may eventually see that $t_n=(\PM)_1u_1+\cdots+(\PM)_nu_n$ for some $u_1,\ldots,u_n\in\T^n(\Ms)$, and consequently $t_n\in S_n$.

\end{proof}


\end{document}